%% file: main.tex
\documentclass[runningheads]{llncs}
\usepackage{packages}
\usepackage{mymacros}
\usepackage{commands}
\usepackage[margin=1in]{geometry}
\begin{document}
\title{
% \sdeni{T\^atonnement beyond Constant Elasticity of Substitution }{
A Consumer-Theoretic Characterization of Fisher Market Equilibria
% }
% \sdeni{}{: Expenditure Minimization to the Rescue}
% \thanks{Supported by organization x.}
}
\author{Denizalp Goktas\inst{1}
% \orcidID{0000-1111-2222-3333} 
\and
Enrique Areyan Viqueira\inst{1}
% \orcidID{1111-2222-3333-4444} 
\and
Amy Greenwald\inst{1}
% \orcidID{2222--3333-4444-5555}
\\
\email{\{denizalp\_goktas, eareyan, amy\_greenwald\}@brown.edu}
}
\authorrunning{Goktas et al.}
\institute{Brown University, 115 Waterman st., Providence RI 02906, USA
% \url{http://www.springer.com/gp/computer-science/lncs}
}
\maketitle  % typeset the header of the contribution
\input{abstract}
%
% Allows equations to break over pages
\allowdisplaybreaks

\input{intro}

\input{prelim}
\input{overview}
% \input{program}
\input{newconvex}
\input{equivalence}
\input{tatonnement}

\input{experiments}
\input{conclusion}

\input{acks}
%\newpage

\printbibliography

\newpage
\appendix
% \section{Section \ref{sec:EG-dual} Omitted Proofs}\label{sec:Proofs3}

\section{Section \ref{sec:program} Omitted Proofs}\label{sec:Proofs4}\label{sec:Proofs3}
\input{appendix/proofs_sec3}

\input{appendix/proofs_sec4}

\section{Danskin's Theorem Section \ref{sec:equiv} Omitted Proofs }\label{sec:Proofs5}
\input{appendix/proofs_sec5}

% \section{Section \ref{sec:convergence} Omitted Proofs}\label{sec:Proofs6}
% \input{appendix/proofs_sec6}

\end{document}

%% file: abstract.tex
\begin{abstract}

In this paper, we bring consumer theory to bear in the analysis of Fisher markets whose buyers have arbitrary continuous, concave, homogeneous (CCH) utility functions representing locally non-satiated preferences.
The main tools we use are the dual concepts of expenditure minimization and indirect utility maximization. 
%we derive the dual of the Eisenberg-Gale program and
First, we use expenditure functions to construct a new convex program whose dual, like the dual of the Eisenberg-Gale program, characterizes the equilibrium prices of CCH Fisher markets. 
%, which also provide further intuition for the Eisenberg-Gale program.
%
%\samy{}{computationally direct} \amy{i don't know what this expression means, ``computationally direct''? let's discuss}\deni{I prefer ``direct computational argument''. As in, we prove things with a series of equalities and not by more verbose proof techniques if that makes sense?}
% \amy{what about just plain direct. i think direct alone can imply a series of equalities.} \deni{Just a quick google, direct seems to mean cause to consequence argument, while indirect means a proof by contradiction, so I think the extra precision is good because \citeauthor{fisher-tatonnement}'s proof is also direct, it is a proof by cases with words but follows the cause to consequence logic}
We then prove that the subdifferential of the dual of our convex program is equal to the negative excess demand in the associated market, which makes generalized gradient descent equivalent to computing equilibrium prices via t\^atonnement.
%\csch{Finally, we use our novel characterization of equilibrium prices via expenditure functions to show that a discrete t\^atonnement process converges at a rate of $O(\nicefrac{1}{t})$ in Fisher markets with continuous, \emph{strictly\/} concave, homogeneous (CSCH) utility functions---a class of utility functions beyond the class of CES utility functions, the largest class for which convergence results were previously known. CSCH Fisher markets include nested and mixed CES Fisher markets, thus providing a meaningful expansion of the space of Fisher markets that is solvable via t\^atonnement.}
Finally, we run a series of experiments which suggest that t\^atonnement may converge at a rate of $O(\nicefrac{(1+E)}{t^2})$ in CCH Fisher markets that comprise buyers with elasticity of demand bounded by $E$.
Our novel characterization of equilibrium prices may provide a path to proving the convergence of t\^atonnement in Fisher markets beyond those in which buyers utilities exhibit constant elasticity of substitution.
%\sdeni{}{one of} the largest classes \sdeni{}{of markets} for which convergence results are known. \amy{this is a weird thing to say, given Fig 1.}
% Finally, we introduce a weaker equilibrium concept we call uniform expected competitive equilibrium and show that t\^atonnement convergence to such an equilibrium at a rate of $O(\nicefrac{1}{\sqrt{t}})$ in all CCH Fisher markets.

\keywords{Market Equilibrium \and Market Dynamics \and Fisher Market.}

%\amy{one convex program, for all CCH utilities. two classes of convergence results. linear, quasilinear: not strictly concave, don't get convergence. leontief, cobb-douglas: strictly concave, get convergence.} 

\end{abstract} 

%% file: intro.tex
\section{Introduction}
\label{sec:intro}

One of the seminal achievements in mathematical economics is the proof of existence of equilibrium prices in \mydef{Arrow-Debreu competitive economies} \cite{arrow1954existence}.
This result, while celebrated, is non-constructive, and thus provides little insight into the computation of equilibrium prices. 
The computational question dates back to L\'eon Walras, a French economist, who in 1874 conjectured that a decentralized
price-adjustment process he called \mydef{t\^atonnement},
which reflects market behavior, would converge to equilibrium prices \cite{walras}.
An early positive result in this vein was provided by Arrow, Block and Hurwicz, who showed that a continuous version of t\^atonnement converges in markets with an aggregate demand function that satisfies the \mydef{weak gross substitutes (WGS)} property \cite{arrow-hurwicz}.
Unfortunately, following this initial positive result, Herbert Scarf provided his eponymous example of an economy for which the t\^atonnement process does not converge, dashing all hopes of the t\^atonnement process justifying the concept of market equilibria in general \cite{scarf-eaves}. 
Nonetheless, further study of t\^atonnement in simpler models than a full-blown Arrow-Debreu competitive economy remains important, as some real-world markets are indeed simpler \cite{duffie1989simple}.

For market equilibria to be justified, not only should they be backed by a natural price-adjustment process such as t\^atonnement, as economists have long argued, they should also be computationally efficient.
As Kamal Jain put it, ``If your laptop cannot find it, neither can the market'' \cite{AGT-book}.
A detailed inquiry into the computational properties of market equilibria was initiated by \citeauthor{devanur2002market} \cite{devanur2002market,devanur2008market}, who studied a special case of the Arrow-Debreu competitive economy known as the \mydef{Fisher market} \cite{brainard2000compute}. This model, for which Irving Fisher computed equilibrium prices using a hydraulic machine in the 1890s, is essentially the Arrow-Debreu model of a competitive economy in which there are no firms, and buyers are endowed with an artificial currency \cite{AGT-book}.
\citeauthor{devanur2002market} \cite{devanur2002market} discovered a connection between the \mydef{Eisenberg-Gale convex program} and Fisher markets in which buyers have linear utility functions, thereby providing a (centralized) polynomial time algorithm for equilibrium computation in these markets~\cite{devanur2002market,devanur2008market}. 

Their work was built upon by \citeauthor{jain2005market} \cite{jain2005market}, who extended the Eisenberg-Gale program to all Fisher markets whose buyers have \mydef{continuous, concave, and homogeneous (CCH)} utility functions.
Further, they proved that the equilibrium of Fisher markets for buyers with CCH utility functions can be computed in polynomial time by interior point methods.%
\footnote{We refer to Fisher markets that comprise buyers with a certain utility function by the name of the utility function, e.g., we call a Fisher market that comprise buyers with CCH utility functions a CCH Fisher market.}
Even more recently, \citeauthor{gao2020first} \cite{gao2020first} go beyond interior point methods to develop algorithms that converge in \mydef{linear}, \mydef{quasilinear}, and \mydef{Leontief} Fisher markets.
%\amy{what about Cobb Douglas?}\deni{We have a closed-form solution for Cobb Douglas}
%\footnote{We refer the the reader to \Cref{sec:prelim} for an explanation of these utility functions.}
However, unlike t\^atonnement, these methods provide little insight into how markets reach equilibria.

More recently, \citeauthor{cole2008fast} \cite{cole2008fast, cole2010discrete}, and \citeauthor{fisher-tatonnement} \cite{fisher-tatonnement} showed the fast convergence of t\^atonnement in Fisher markets where the buyers' utility functions satisfy weak gross substitutes with bounded elasticity of demand, and the \mydef{constant elasticity of substitution (CES)} properties respectively, the latter of which is a subset of the class of CCH utility functions \cite{fisher-tatonnement, cole2008fast, cole2010discrete}.
%Although t\^atonnement has also been criticized for being a centralized process 
Aside from t\^atonnement being a plausible model of real-world price movements due to its decentralized nature, \citeauthor{cole2008fast} argue for the plausibility of t\^atonnement by proving that it is an abstraction for in-market processes in a real-world-like model called the ongoing market model \cite{cole2008fast, cole2010discrete}.
The plausibility of t\^atonnement as a natural price-adjustment process has been further supported by \citeauthor{gillen2020divergence} \cite{gillen2020divergence}, who demonstrated the predictive accuracy of t\^atonnement in off-equilibrium trade settings \cite{gillen2020divergence}.
This theoretical and empirical evidence for t\^atonnement makes it even more important to understand its convergence properties,
%for as large a class of competitive economies as possible, 
so that we can better characterize those markets for which we can predict price movements and, in turn, equilibria. %\sdeni{}{A more detailed account of related works can be found in the full version of this paper.}
%---beyond Fisher markets with CES utilities \cite{fisher-tatonnement}. %\amy{what about \cite{cole2008fast}? i.e., what about WGS markets?} \deni{\cite{cole2008fast, cole2010discrete} are discussing WGS markets and t\^atonnement representing in-market processes.}

\wine{Another price-adjustment process that has been shown to converge to market equilibria in Fisher markets is \mydef{proportional response dynamics}, first introduced by \citeauthor{first-prop-response} for linear utilities \cite{first-prop-response}; then expanded upon and shown to converge by \citeauthor{proportional-response} for all CES utilties \cite{proportional-response}; and very recently shown to converge in Arrow-Debreu exchange markets with linear utilities by \citeauthor{branzei2021proportional} \cite{branzei2021proportional}.
The study of the proportional response process was proven fundamental when \citeauthor{grad-prop-response} \cite{grad-prop-response} noticed its relationship to gradient descent.
This discovery opened up a new realm of possibilities in analyzing the convergence of market equilibrium processes.
For example, it allowed \citeauthor{cheung2018dynamics} \cite{cheung2018dynamics} to generalize the convergence results of proportional response dynamics to Fisher markets for buyers with mixed CES utilities.
This same idea was applied by \citeauthor{fisher-tatonnement} \cite{fisher-tatonnement} to prove the convergence of t\^atonnement in Leontief Fisher markets, using the equivalence between generalized gradient descent 
on the dual of the Eisenberg-Gale program 
and t\^atonnement, first observed by \citeauthor{devanur2008market} \cite{devanur2008market}}.

\paragraph{Our Approach and Findings}
In consumer theory \cite{mas-colell}, consumers/buyers are assumed to solve the \mydef{utility maximization problem} (\mydef{UMP}), in which each buyer maximizes its utility constrained by its budget, thereby discovering its optimal demand.
Dual to this problem is the \mydef{expenditure minimization problem} (\mydef{EMP}), in which each buyer minimizes its expenditure constrained by its desired utility level, an alternative means of discovering its optimal demand.
These two problems are intimately connected by a deep mathematical structure,
%, which to our knowledge has not yet been exploited in the computation of market equilibria.
%On the contrary, 
yet most existing approaches to computing market equilibria focus on UMP only.

%It is well known that the Eisenberg-Gale convex program captures through its primal the equilibrium allocations of goods to buyers, and through its dual the equilibrium prices of Fisher markets, assuming buyers with CCH utility functions \cite{jain2005market}.

% \amy{move at least some of this to new overleaf}
% \sdeni{We propose a program which generalizes the Eisenberg-Gale program from CCH utility functions which represent locally non-satiated preferences to all CCH utility functions. Our program provides a simple convex programming solution to Fisher markets with CCH utility function which represent locally satiated preferences. Examples of such utility functions were first introduced by \citeauthor{cole2016convex} \cite{cole2016convex} and algorithms to solve some variants of Fisher markets where buyers have some variants of linear utility functions which represent locally satiated preferences were later developed by \citeauthor{bei2017spending} \cite{bei2017spending}. Our program also generalizes a convex program proposed by \citeauthor{devanur2009fisher} \cite{devanur2009fisher} for quasilinear Fisher markets, for which the Eisenberg-Gale program fails since quasilinear utilities are locally satiated.}{}

In this paper, 
%we derive the dual of the Eisenberg-Gale program and
we exploit the relationship between EMP and equilibrium prices to provide a new convex program, which like the seminal Eisenberg-Gale program characterizes the equilibrium prices of Fisher markets
%via expenditure functions, 
assuming buyers with arbitrary CCH utility functions.
Additionally, by exploiting the duality structure between UMP and EMP, we provide a straightforward interpretation of the dual of our program, which also sheds light on the dual of Eisenberg-Gale program.
In particular, while it is known that an equilibrium allocation that solves the Eisenberg-Gale program is one that maximizes the buyers' utilities given their budgets at equilibrium prices (UMP; the primal), we show that equilibrium prices are those that minimize the buyers' expenditures at the utility levels associated with their equilibrium allocations (EMP; the dual).
%, thereby providing intuition for the dual of the Eisenberg-Gale program, which was heretofore difficult to interpret.

Our characterization of CCH Fisher market equilibria via UMP and EMP also allows us %\amy{is it really our characterization that allows us to do this?}\deni{Yes, because we use shepherd's lemma} 
to prove that the subdifferential of the dual of our convex program is equal to the negative excess demand in the corresponding market \cite{fisher-tatonnement}.%
\footnote{Similarly, it is known that the subdifferential of the dual of the Eisenberg-Gale program is equal to the negative excess demand in the corresponding market \cite{fisher-tatonnement}.  Our result
%\samy{, whose proof is a straightforward step-by-step computation,}{}
also implies this known result, since the two programs' objective functions differ only by a constant.}
%(\Cref{excess-demand})
Consequently, solving the dual of our convex program via generalized gradient descent is equivalent to t\^atonnement (just as generalized gradient descent 
on the dual of the Eisenberg-Gale program is equivalent to t\^atonnement \cite{devanur2008market}).

%\csch{This result, combined with our understanding of the dual of our convex program via EMP, allows us to prove that the t\^atonnement process converges at the sublinear rate of $O(\nicefrac{1}{t})$ in all Fisher markets where buyers have continuous, \emph{strictly\/} concave, and homogeneous (\mydef{CSCH}) utility functions. %(\Cref{convergence-result}). This class of markets encompasses Fisher markets whose buyers have CES utility functions, the largest class for which convergence results were previously known \cite{fisher-tatonnement}. It also includes important additional classes of Fisher markets, such as \mydef{mixed CES markets} \cite{cheung2018dynamics}, i.e., markets in which buyers have CES utilities with different $\rho$ parameters, and \mydef{nested CES markets} \cite{jain2006equilibria}, i.e.. markets in which buyers have different elasticities of substitution between different goods.}

%This result, combined with our understanding of the dual of our convex program via EMP, may provide a path to generalizing the convergence of t\^atonnement beyond CES Fisher markets.

Finally, we run a series of experiments which suggest \amy{conjecture!} that t\^atonnment may converge at a rate of $O(\nicefrac{(1 + E)}{t^2})$ in CCH markets where buyers have \mydef{bounded elasticity of demand (BED)} with elasticity parameter $E$, a class of markets that includes CES Fisher markets. %\sdeni{}{one of} the largest classes \sdeni{}{of markets} for which convergence rates are known. \amy{this is a weird thing to say, given Fig 1.}
Assuming bounded elasticity of demand, bounded changes in prices result in bounded changes in demand.
%\samy{ensuring a steady progress toward equilibrium}{} \amy{too strong, i think!}\deni{I think we should delete but I think this is correct because it is essentially saying that the demand is Lipschtiz smooth!}
A summary of all known t\^atonnement convergence rate results, as well as this conjecture, appears in \Cref{fig:utility-functions}.

\paragraph{Roadmap}
In Section \ref{sec:prelim}, we introduce essential notation and definitions, and summarize our results.
% In Section \ref{sec:EG-dual}, .
In Section~\ref{sec:program}, we derive the dual of the Eisenberg-Gale program and propose a new convex program whose dual characterizes equilibrium prices in CCH Fisher markets via expenditure functions.
In \Cref{sec:equiv}, we show that the subdifferential of the dual of our new convex program is equivalent to the negative excess demand in the market, which implies an equivalence between generalized gradient descent and t\^atonnement.
In Section \ref{sec:convergence}, we include an empirical analysis of t\^atonnement in CCH Fisher markets.

\input{figure1}

%% file: figure1.tex
\sidecaptionvpos{figure}{c}
\begin{SCfigure}
\centering
\begin{tikzpicture}[framed, scale=0.3, transform shape]
          \node [venn circle = cyan, minimum size=2cm, scale= 2] (CES) at (5,-2) {\begin{tabular}{c}
               {CES ($\rho < 1$)} \\
               {$(1-\Theta(1))^t$}
          \end{tabular} };
          
          \node [venn circle = cyan, minimum size=2cm, scale=1.7] (leontief) at (7,2) {\begin{tabular}{c}
               {Leontief} \\
               {$O(\nicefrac{1}{t})$}
          \end{tabular}};

          \node [venn circle = magenta, minimum size= 15cm] (CCH) at (5,0) {};
          \node [label, scale = 2] (CCHlabel) at (5,5) {\begin{tabular}{c}
               {CCH BED} \\
               {$O(\nicefrac{(1 + E)}{t^2})$}
          \end{tabular}};  
          \node [venn circle = blue, minimum size=10cm] (WGS) at (-0.5,0.5) {};
          \node [label, scale = 1.5] (WGSlabel) at (-3.9, 0.75) {\begin{tabular}{c}
               {WGS} \\
               {$\left(1-\Theta(1)\right)^t$}
          \end{tabular}};

          \node [label] (title) at (-2,8){          \begin{tabular}{c}
               {\Huge Space of Fisher} \\
                {\Huge markets}
          \end{tabular}};
        %   \node [label, rotate=65] (WGSCES) at (3.25, -1.25) {CES with $\rho > 0$};
          \node [label, scale = 2] (WGSCES) at (-3, -8) {CES with $\rho > 0$};
          \draw[->] (WGSCES) -- (3.25, -1.25);
\end{tikzpicture}
\caption{
    % \scriptsize
    The convergence rates of t\^atonnement in different Fisher markets.
    We color previous contributions blue, and our conjecture in red.
    %, i.e., we study \sdeni{}{CCH Fisher markets with bounded elasticity of substitution.
    We note that the convergence rate for WGS markets does not apply to markets where the elasticity of demand is unbounded, e.g., linear Fisher markets; likewise, the convergence rate for CES Fisher markets does not apply to linear Fisher markets.}
    \label{fig:utility-functions}
\end{SCfigure}
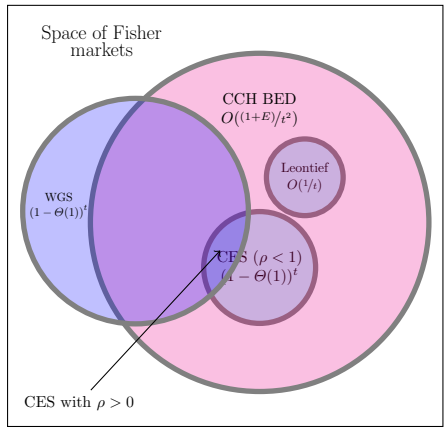

%% file: prelim.tex
\section{Preliminaries and an Overview of Results}
\label{sec:prelim}

We use Roman uppercase letters to denote sets (e.g., $X$), bold uppercase letters to denote matrices (e.g., $\allocation$), bold lowercase letters to denote vectors (e.g., $\price$), and Roman lowercase letters to denote scalar quantities (e.g., $c$). We denote the $i^{th}$ row vector of any matrix (e.g., $\allocation$) by the equivalent bold lowercase letter with subscript $i$ (e.g., $\allocation[\buyer])$. Similarly, we denote the $j$th entry of a vector (e.g., $\price$ or $\allocation[\buyer]$) by the corresponding Roman lowercase letter with subscript $j$ (e.g., $\price[\good]$ or $\allocation[\buyer][\good]$).
We denote the set of numbers $\left\{1, \hdots, n\right\}$ by $[n]$, the set of natural numbers by $\N$, the set of real numbers by $\R$, the set of non-negative real numbers by $\R_+$ and the set of strictly positive real numbers by $\R_{++}$. We denote by $\project[X]$ the Euclidean projection operator onto the set $X \subset \R^{n}$: i.e., $\project[X](\x) = \argmin_{\z \in X} \left\| \x - \z \right\|_2$. 
We also define some set operations. Unless otherwise stated, the sum of a scalar by a set and of two sets is defined as the Minkowski sum, \wine{e.g., $c + A = \{c + a \mid a \in A\}$ and $A+B = \{a + b \mid a \in A, b \in B\}$,} and the product of a scalar by a set and two sets is defined as the Minkowski product\wine{, e.g., $cA = \{ca \mid a \in A \}$ and $AB = \{ab \mid a \in A, b \in B\}$}{.}

\subsection{Consumer Theory}

In this paper, we consider the general class of utility functions $\util[\buyer]: \R^\numgoods \to \R$ that are continuous, concave and homogeneous.
\wine{Recall that a set $U$ is open if for all $x \in U$ there exists an $\epsilon > 0$ such that the open ball $B_{\epsilon}(x)$ centered at $x$ with radius $\epsilon$ is a subset of $U$, i.e., $B_{\epsilon}(x) \subset U$.
A function $f: \R^m \to \R$ is said to be \mydef{continuous} if $f^{-1}(U)$ is open for every $U \subset \R$.
A function $f: \R^m \to \R$ is said to be concave if $\forall \lambda \in (0, 1), \x, \y \in \R^m, f(\lambda \x + (1-\lambda)\y) \geq \lambda f(\x) +(1-\lambda)f(\y)$ and \mydef{strictly concave} if strict inequality holds.
A function $f: \R^m \to \R$ is said to be \mydef{homogeneous} of degree $k \in \N_+$ if for all $\allocation[ ][] \in \R^{m}, \lambda > 0, f(\lambda \allocation[ ]) = \lambda^{k} f(\allocation[ ][])$.
Unless otherwise indicated, without loss of generality a homogeneous function is assumed to be homogeneous of degree 1.%
\footnote{A homogeneous utility function of degree $k$ can be made homogeneous of degree 1 without affecting the underlying preference relation by passing the utility function through the monotonic transformation $x \mapsto \sqrt[k]{x}$.}
A utility function $\util[ ]: \R^{\numgoods} \to \R$ assigns a real value to elements of $\R^{\numgoods}$, i.e., to every possible allocation of goods.
% \sdeni{}{Let $\numgoods \in \N$ be the number of goods, we define utility functions as functionals on $\R^{\numgoods + 1}_+$ such that the $(\numgoods +1)^{th}$ dimension represents money (or the artifical currency). That is, a utility function $\util[ ]: \R^{\numgoods}_+ \times \R_+ \to \R$ assigns a real value to elements of $\R^{\numgoods + 1} $, i.e., to every possible allocation of goods and money. Note that this definition is more general than the definition of utility functions often used in the Fisher market literature, and simply implies that buyers can attribute value to money}.\footnote{\sdeni{}{A broader discussion of utility functions that attribute value to money can be found in \cref{sec:money}. Note that this definition is also in line with Fisher markets as Fisher markets are a special case of an Arrow-Debreu market where money is a good in the market which only buyers are endowed with.}} 
Every utility function then represents some \mydef{preference relation} $\succeq$ over goods such that if for two bundle of goods $\x, \y \in \R^\numgoods$, $\util[ ](\x) \geq \util[ ](\y)$ then $\x \succeq \y$.
A preference relation $\succeq$ is said to be \mydef{locally non-satiated} iff for all $\x \in \R^\numgoods$ and $\epsilon > 0$, there exists $\y \in B_{\epsilon}(\x)$ such that $\y \succ \x$.
\emph{Throughout this paper, we assume utility functions represent locally non-satiated preferences.}
If a buyer's utility function represents locally non-satiated preferences, there always exists a better bundle for that buyer if their budget increases.}
\wine{We define some important concepts that pertain to utility functions.}
The \mydef{indirect utility function} $\indirectutil[\buyer]: \mathbb{R}^{\numgoods}_+ \times \mathbb{R}_+ \to \mathbb{R}_+$ takes as input prices $\price$ and a budget $\budget[\buyer]$ and outputs the maximum utility the buyer can achieve at those prices given that budget, i.e., $ \indirectutil[\buyer](\price, \budget[\buyer]) = \max_{\allocation[ ] \in \mathbb{R}^\numgoods_+ : \price \cdot \allocation[ ] \leq \budget[\buyer]} \util[\buyer](\allocation[ ])$ 
\wine{If the utility function is continuous, then the indirect utility function is continuous and homogeneous of degree 0 in $\price$ and $\budget[\buyer]$ jointly,
i.e., $\forall \lambda > 0, \indirectutil[\buyer] (\lambda \price, \lambda \budget[\buyer]) = 
\indirectutil[\buyer] (\price, \budget[\buyer])$ is non-increasing in $\price$, strictly increasing in $\budget[\buyer]$, and convex in $\price$ and $\budget[\buyer]$.}

The \mydef{Marshallian demand} is a correspondence $\marshallian[\buyer]: \R^{\numgoods}_+ \times \R_+ \rightrightarrows \R^{\numgoods}_+$ that takes as input prices $\price$ and a budget $\budget[\buyer]$ and outputs the utility maximizing allocation of goods at budget $\budget[\buyer]$, i.e., $ \marshallian[\buyer](\price, \budget[\buyer]) = \argmax_{\allocation[ ] \in \mathbb{R}^\numgoods_+ : \price \cdot \allocation[ ] \leq \budget[\buyer]} \util[\buyer](\allocation[ ])$:
% \wine{}{$\marshallian[\buyer](\price, \budget[\buyer]) = \argmax_{\allocation[ ] \in \mathbb{R}^\numgoods_+} \max_{\money[ ] \in \R_+: \price \cdot \allocation[ ] + \money[ ] \leq \budget[\buyer]} \util[\buyer](\allocation[ ], \money[ ])$.} 
i.e., $\marshallian[\buyer](\price, \budget[\buyer]) = \argmax_{\allocation[\buyer] \in \mathbb{R}^\numgoods_+ : \price \cdot \allocation[\buyer] \leq \budget[\buyer]} \util[\buyer](\allocation[\buyer])$.
\wine{The Marshallian demand is convex-valued if the utility function is continuous and concave, and unique if the utility function is continuous and strictly concave}

The \mydef{expenditure function} $\expend[\buyer]: \mathbb{R}^{\numgoods}_+ \times \mathbb{R}_+ \to \mathbb{R}_+$ takes as input prices $\price$ and a utility level $\goalutil[\buyer]$ and outputs the minimum amount the buyer must spend to achieve that utility level at those prices, i.e., $\expend[\buyer](\price, \goalutil[\buyer]) = \min_{\allocation[ ] \in \mathbb{R}^\numgoods_+: \util[\buyer](\allocation[ ]) \geq \goalutil[\buyer]} \price \cdot \allocation[ ]$ 
% \sdeni{}{$\expend[\buyer](\price, \goalutil[\buyer])  = \min_{\allocation[ ] \in \mathbb{R}^\numgoods_+, \money[ ] \in \R_+: \util[\buyer](\allocation[ ], \money[ ]) \geq \goalutil[\buyer]} \price \cdot \allocation[\buyer] + \money[ ]$}.
\wine{If the utility function $\util[\buyer]$ is continuous, then the expenditure function is continuous and homogeneous of degree 1 in $\price$ and $\goalutil[\buyer]$ jointly,
$\forall \lambda >0, \expend[\buyer] (\lambda \price, \lambda \goalutil[\buyer]) = \lambda \expend[\buyer] (\price, \goalutil[\buyer])$ is non-decreasing in $\price$, strictly increasing 
in $\goalutil[\buyer]$, and concave in $\price$ and $\goalutil[\buyer]$.}{}

The \mydef{Hicksian demand} is a correspondence $\hicksian[\buyer]: \mathbb{R}^{\numgoods}_+ \times \mathbb{R}_+  \rightrightarrows \mathbb{R}_+$ that takes as input prices $\price$ and a utility level $\goalutil[\buyer]$ and outputs the cost-minimizing allocation of goods at utility level $\goalutil[\buyer]$,  i.e., $\hicksian[\buyer](\price, \goalutil[\buyer]) = \argmin_{\allocation[ ] \in \mathbb{R}^\numgoods_+: \util[\buyer](\allocation[ ]) \geq \goalutil[\buyer]} \price \cdot \allocation[ ]$
% \sdeni{}{$\hicksian[\buyer](\price, \goalutil[\buyer])  = \argmin_{\allocation[ ] \in \mathbb{R}^\numgoods_+} \min_{\money[ ] \in \R_+: \util[\buyer](\allocation[ ], \money[ ]) \geq \goalutil[\buyer]} \price \cdot \allocation[\buyer] + \money[ ]$}.
% If the utility function is continuous, then the Hicksian demand is homogeneous of degree 1 in $\price$, and is a convex set which is unique iff the the utility function of the buyer is strictly concave
\wine{The Hicksian demand is convex-valued if the utility function is continuous and concave, and unique if the utility function is continuous and strictly concave \cite{levin-notes, mas-colell}.}{}

In consumer theory, the demand of buyers can be determined by studying two dual problems, the \mydef{utility maximization problem (UMP)} and \mydef{the expenditure minimization problem (EMP)}. The UMP refers to the buyer's problem of maximizing its utility constrained by its budgets (i.e., optimizing its indirect utility function) in order to obtain its optimal demand (i.e., Marshallian demand), while the EMP refers to the buyer's problem of minimizing its expenditure constrained by its desired utility level (i.e., optimizing its expenditure function) in order to obtain its optimal demand (i.e., Hicksian demand). When the utilities are continuous, concave and represent locally non-satiated preferences the UMP and EMP are related through the following identities, which we use throughout the paper:
\begin{align}
    &\forall \budget[\buyer] \in \mathbb{R}_{+} & \expend[\buyer](\price, \indirectutil[\buyer](\price, \budget[\buyer])) = \budget[\buyer]\label{expend-to-budget}\\
    &\forall \goalutil[\buyer] \in \mathbb{R}_{+} &  \indirectutil[\buyer](\price, \expend[\buyer](\price, \goalutil[\buyer])) = \goalutil[\buyer]\label{indirect-to-value}\\
    & \forall \budget[\buyer] \in \mathbb{R}_{+} & \hicksian[\buyer](\price, \indirectutil[\buyer](\price, \budget[\buyer])) = \marshallian[\buyer](\price, \budget[\buyer])
    \label{hicksian-marshallian}\\
    &\forall \goalutil[\buyer] \in \mathbb{R}_{+} &  \marshallian[\buyer](\price, \expend[\buyer](\price, \goalutil[\buyer])) = \hicksian[\buyer](\price, \goalutil[\buyer])\label{marshallian-hicksian}
\end{align}

A good $\good \in \goods$ is said to be a \mydef{gross substitute (complement)} for a good $k \in \goods \setminus \{\good\}$ if $\sum_{\buyer \in \buyers} \marshallian[\buyer][\good](\price, \budget[\buyer])$ is increasing (decreasing) in $\price[k]$.
If the aggregate demand, $\sum_{\buyer \in \buyers} \marshallian[\buyer][\good](\price, \budget[\buyer])$, for good $k$ is instead weakly increasing (decreasing), good $\good$ is said to be a \mydef{weak gross substitute} (\mydef{complement}) for good $k$.

\if 0
The class of homogeneous utility functions includes the well-known \mydef{linear}, \mydef{Cobb-Douglas}, and \mydef{Leontief utility} functions:
\begin{align}
    &\text{Linear:} & \util[\buyer](\allocation[\buyer]) = \sum_{\good \in \goods} \valuation[\buyer][\good] \allocation[\buyer][\good]\\
    &\text{Cobb-Douglas:} & \util[\buyer](\allocation[\buyer]) = \prod_{\good \in \goods} \allocation[\buyer][\good]^{\valuation[\buyer][\good]}\\
    &\text{Leontief:} & \util[\buyer](\allocation[\buyer]) = \min_{\good : \valuation[\buyer][\good] \neq 0} \frac{\allocation[\buyer][\good]}{\valuation[\buyer][\good]}
\end{align}

\noindent
where each utility function is parameterized by the vector of valuations $\valuation[\buyer] \in \mathbb{R}_+^{\numbuyers}$, where each $\valuation[\buyer][\good]$ quantifies the value of good $j$ to buyer $i$.%
\footnote{For certain utility functions, such as quasilinear utilities, the vector of valuations captures the monetary value of goods.}
%which quantifies the preference of buyer $\buyer$ over goods.
\fi

The class of homogeneous utility functions includes the well-known \mydef{linear}, \mydef{Cobb-Douglas}, and \mydef{Leontief} utility functions, each of which is a special case of the \mydef{Constant Elasticity of Substitution (CES)} utility function family, parameterized by $-\infty \leq \rho \leq 1$, and
given by
%\begin{align}
$\util[\buyer](\allocation[\buyer]) = \sqrt[\rho]{ \sum_{\good \in \goods} \valuation[\buyer][\good] \allocation[\buyer][\good]^\rho}$.
%\enspace .
%\end{align}
Linear utility functions are obtained when $\rho$ is $1$, while Cobb-Douglas and Leontief utility functions are obtained when $\rho$ approaches $0$ and $-\infty$, respectively. For $0 < \rho \leq 1$, goods are gross substitutes, e.g, Sprite and Coca-Cola, for $\rho = 1$; goods are perfect substitutes, e.g., Pepsi and Coca-Cola; and for $\rho < 0$, goods are complementary, e.g., left and right shoes.

The \mydef{(price) elasticity of demand} reflect how demand varies in response to a change in price.
More specifically, buyer $\buyer$'s elasticity of demand for good $\good \in \goods$ with respect to the price of good $k \in \goods$ is defined as $\frac{\partial \marshallian[\buyer][\good](\price, \budget[\buyer])}{\partial \price[k]} \frac{\price[k]}{\marshallian[\buyer][\good](\price, \budget[\buyer])}$. A buyer is said to have \mydef{bounded elasticity of demand} with elasticity parameter $E$ if $\min_{\price \in \R^\numgoods_+, \good, k \in \goods } \left\{ - \frac{\partial \marshallian[\buyer][\good](\price, \budget[\buyer])}{\partial \price[k]}  \frac{\price[k]}{\marshallian[\buyer][\good](\price, \budget[\buyer])} \right\} = E < \infty$.

% Another type of homogeneous utility function, commonly studied in the mechanism design literature \cite{cole2016convex}, is the \mydef{quasilinear} utility function, which is parameterized by a vector of valuations $\valuation[\buyer] \in \mathbb{R}^{\numgoods}$  and a vector of prices $\price \in \mathbb{R}^{\numgoods}$, and given by $\util[\buyer](\allocation[\buyer]; \price) = \allocation[\buyer] \cdot (\valuation[\buyer] - \price)$.

\subsection{Fisher Markets}

A \mydef{Fisher market} comprises $\numbuyers$ buyers and $\numgoods$ divisible goods \cite{brainard2000compute}. As is usual in the literature, we assume that there is one unit of each good available \cite{AGT-book}. Each buyer $\buyer \in \buyers$ has a budget $\budget[\buyer] \in \mathbb{R}_{+}$ and a utility function $\util[\buyer]: \mathbb{R}_{+}^{\numgoods} \to \mathbb{R}$. An instance of a Fisher market is thus given by a tuple $(\numbuyers, \numgoods, \util, \budget)$ where $\util = \left\{\util[1], \hdots, \util[\numbuyers] \right\}$ is a set of utility functions, one per buyer, and $\budget \in \R_{+}^{\numbuyers}$ is the vector of buyer budgets. We abbreviate as $(\util, \budget)$ when $\numbuyers$ and $\numgoods$ are clear from context.

When the buyers' utility functions in a Fisher market are all of the same type, we qualify the market by the name of the utility function, e.g., a Leontief Fisher market.
%Further, we refer to markets which comprise buyers with CCH utility functions that represent locally non-satiated preferences as \mydef{homothetic markets}.
Considering properties of goods, rather than buyers, a (Fisher) market satisfies \mydef{gross substitutes} (resp. \mydef{gross complements}) if all pairs of goods in the market are gross substitutes (resp. gross complements).
A Fisher market is \mydef{mixed} if all pairs of goods are either gross complements or gross substitutes.
A Fisher market exhibits \mydef{bounded elasticity of demand} with parameter $E$, if the elasticity of demand of the buyer with highest elasticity of demand is $E < \infty$.
% \samy{We refer the reader to \Cref{fig:utility-functions} for a summary of the relationships among various Fisher markets.}{}

An \mydef{allocation} $\allocation$ is a map from goods to buyers, represented as a matrix s.t. $\allocation[\buyer][\good] \ge 0$ denotes the quantity of good $\good \in \goods$ allocated to buyer $\buyer \in \buyers$. Goods are assigned \mydef{prices} $\price \in \mathbb{R}_+^{\numgoods}$. A tuple $(\allocation^*, \price^*)$ is said to be a \mydef{competitive (or Walrasian) equilibrium} of Fisher market $(\util, \budget)$ if 1.~buyers are utility maximizing constrained by their budget, i.e., for all $\buyer \in \buyers, \allocation[\buyer]^* \in 
%\argmax_{\allocation[ ] : \allocation[ ] \cdot \price^* \leq \budget[\buyer]} \util[\buyer](\allocation[ ])
\marshallian[\buyer] (\price^*, \budget[\buyer])$;
and 2.~the market clears, i.e., for all $\good \in \goods,  \price[\good]^* > 0$ implies $\sum_{\buyer \in \buyers} \allocation[\buyer][\good]^* = 1$; and $\price[\good]^* = 0$ implies $\sum_{\buyer \in \buyers} \allocation[\buyer][\good]^* \leq 1$.

%\subsection{Computation of Fisher Market Equilibria}

% If $(\util, \budget)$ is a quasilinear Fisher market, then the optimal solutions, $\price^*$ and $\allocation^*$, to the primal and dual of \mydef{Devanur's program} \cite{devanur2009fisher}, respectively, constitute an equilibrium of the Fisher market:
% %
% \begin{equation*}
% \begin{aligned}[c]
%     &\text{\textbf{Primal}} \\ &\min_{\price \in \R^\numgoods_+, \bm{\beta} \in \R_+^\numbuyers } & \sum_{\good \in \goods} \price[\good] - \sum_{\buyer \in \buyers} \budget[\buyer] \log{\beta_{\buyer}} \\
%     &\forall \buyer \in \buyers, \good \in \goods & \price[\good] \geq \valuation[\buyer][\good] \beta_{\buyer} \\
%     &\forall \buyer \in \buyers & \beta_{\buyer} \leq 1 \\
% \end{aligned}
% \quad \vline \quad
% \begin{aligned}[c]
%     &\text{\textbf{Dual}} \\ &\max_{\allocation \in \R^{\numbuyers \times \numgoods}_+, \bm{y} \in \R_+^\numbuyers} & \sum_{\buyer \in \buyers} \budget[\buyer] \log{\util[\buyer]} - y_{\buyer} \\
%     & \forall \buyer \in \buyers  & \util[\buyer] \leq \sum_{\good \in \goods} \valuation[\buyer][\good] \allocation[\buyer][\good] + y_{\buyer}  \\
%     & \forall \good \in \goods & \sum_{\buyer \in \buyers} \allocation[\buyer][\good] \leq 1 \\
% \end{aligned}
% \end{equation*}

If $(\util, \budget)$ is a CCH Fisher market, 
%i.e., $\util$ is a vector of CCH utilities that represent locally non-satiated preferences, 
then the optimal solution $\allocation^*$ to the \mydef{Eisenberg-Gale program}
constitutes an equilibrium allocation, and the optimal solution to the Lagrangian that corresponds to the allocation constraints (\Cref{feasibility-constraint-EG}) are the corresponding equilibrium prices \cite{devanur2002market, cole2019balancing, eisenberg1959consensus}:
\begin{align}
\text{\textbf{Primal}} \notag \\
    & \max_{\allocation \in \R_+^{\numbuyers \times \numgoods}} & \sum_{\buyer \in \buyers} \budget[\buyer] \log{\left(\util[\buyer](\allocation[\buyer])\right)} \\
    & \text{subject to} & \sum_{\buyer \in \buyers} \allocation[\buyer][\good] & \leq 1 &&\forall \good \in \goods\label{feasibility-constraint-EG}
    % && \allocation[\buyer][\good] &\geq 0 &&\forall \buyer \in \buyers, \good \in \goods 
\end{align}

% \amy{this commentary does not really belong here. maybe just delete, since we say this later, and i don't want to harp on it, since their dual is unpublished. but i do think we might want to include their dual here, otherwise Example 1 makes no sense. so please include their dual, either in its own paragraph right here, or as a second column next to the EG primal, with footnotes/caveats that say it is not quite the dual (as we show later).}
% \sdeni{}{We note that \citeauthor{cole2016convex} \cite{cole2016convex} provide dual formulations of the Eisenberg-Gale program for linear and Leontief utilities \cite{cole2016convex}, and in unpublished work, \citeauthor{cole2019balancing} \cite{cole2019balancing} present a generalization of these duals for arbitrary CCH utility functions. However, as we show in \Cref{dual-diff-cole}, the optimal value of the objective of the Eisenberg-Gale primal is not equal to the optimal value of the dual provided by \citeauthor{cole2019balancing} \cite{cole2019balancing}, and hence cannot be the dual of the Eisenberg-Gale program for which strong duality holds.}

% \amy{after inserting the dual, can delete this next sentence, i think:}
% One can then deduce the equilibrium prices, $\price^*$, from the equilibrium allocations, $\allocation^*$, using the Karush–Kuhn–Tucker (KKT) conditions, since they correspond to the dual variables associated with Constraint (\ref{feasility-constraint-EG}).

We define the \mydef{excess demand} correspondence  $\excess: \mathbb{R}^{\numgoods} \rightrightarrows \mathbb{R}^{\numgoods}$, of a Fisher market $(\util, \budget)$, which takes as input prices and outputs a set of excess demands at those prices, as the difference between the demand for each good and the supply of each good:
$
    \excess(\price) = \sum_{\buyer \in \buyers} \marshallian[\buyer](\price, \budget[\buyer]) - \bm{1}_{\numgoods} 
$.

\noindent
where $\bm{1}_{\numgoods}$ is the vector of ones of size $\numgoods$.
% and $\sum_{\buyer \in \buyers} \marshallian[\buyer](\price, \budget[\buyer]) - \bm{1}_{\numgoods} 
% % = \left\{\x - \bm{1}_{\numgoods} \mid \x \in \sum_{\buyer \in \buyers}\marshallian[\buyer](\price, \budget[\buyer]) \right\}
% $

The \mydef{discrete t\^atonnement process} for Fisher markets is a decentralized, natural price adjustment, defined as:
\begin{align}
    \price(t+1) = \price(t) + G(\subgrad(t)) && \text{for } t = 0, 1, 2, \hdots  \label{tatonnement}\\
    \subgrad(t) \in \excess(\price(t))\\
    \price(0) \in \mathbb{R}^{\numgoods}_+\label{tatonnement2} \enspace ,
\end{align}

\noindent
where $G: \mathbb{R}^{\numgoods} \to \mathbb{R}^{\numgoods}$ is a \wine{coordinate-wise} monotonic function s.t.\
% \amy{need to add for all $\bm{x} \in \mathbb{R}^m$}
for all $\good \in \goods, \x, \y \in \R^m$, if $x_{\good} \geq  y_{\good}$, then $G_{\good}(\bm{x}) \geq G_{\good}(\bm{y})$. 
Intuitively, t\^atonnement is an auction-like process in which the seller of $\good \in \goods$ increases (resp. decreases) the price of a good if the demand (resp. supply) is greater than the supply (resp. demand).

\subsection{Subdifferential Calculus and Generalized Gradient Descent}

%\subsubsection{Subgradients and Subdifferentials}

We say that a vector $\subgrad \in \mathbb{R}^{n}$ is a \mydef{subgradient} of a continuous function $f: U \to \mathbb{R}$ at $\bm{a} \in U$ if for all $\bm{x} \in U$,
%
%\begin{align}
$f(\bm{x}) \geq f(\bm{a}) + \subgrad^T (\bm{x} - \bm{a})$.
%\enspace .
%\end{align}
%
%\noindent
The set of all subgradients $\subgrad$ at a point $\bm{a} \in U$ for a function $f$ is called the \mydef{subdifferential} and is denoted by $\subdiff[\bm{x}] f(\bm{a}) = \{\subgrad \mid f(\bm{x}) \geq f(\bm{a}) + \subgrad^T (\bm{x} - \bm{a}) \}$.
If $f$ is convex, then its subdifferential exists everywhere.
If additionally, $f$ is differentiable at $\bm{a}$, so that its subdifferential is a singleton at $\bm{a}$, then the subdifferential at $\bm{a}$ is equal to the gradient.
In this case, \wine{for notational simplicity,}
%if at any point $\bm{a} \in U$, $\subdiff[\bm{x}] f(\bm{a}) = \{ \subgrad \}$, i.e., if $f$ is differentiable at $\a$, then
we write $\subdiff[\bm{x}] f(\bm{a}) = \subgrad$;
in other words, 
%if the subdifferential at a point is a singleton, 
we take the subdifferential to be vector-valued rather than set-valued.
\wine{When both $f$ and $\tilde{f}: U \to \mathbb{R}$ are continuous and convex, subdifferentials satisfy the additivity property:
%
%\begin{align}
$\subdiff[\bm{x}] (f + \tilde{f})(\bm{a}) = \subdiff[\bm{x}] f(\bm{a}) + \subdiff[\bm{x}] \tilde{f}(\bm{a})$.
%\enspace ,
%\end{align}
%
%\noindent
If $f$ is again continuous and convex, and if $\tilde{f}: \mathbb{R} \to \mathbb{R}$ is continuous, convex, and differentiable with derivative $\tilde{f}'$, subdifferentials satisfy the composition property:
%
%\begin{align}
%\enspace ,
%\end{align}
%
%\noindent
$\subdiff[\bm{x}] (\tilde{f} \circ f(\bm{a})) = \tilde{f}'(f(\a)) \cdot \subdiff[\bm{x}] f (\bm{a})$.}

%\subsubsection{Subgradient Methods}

Consider the optimization problem
%\begin{align}
$\min_{\bm{x} \in V} f(\bm{x})\label{optimization-problem}$,
%\enspace ,
%\end{align}
where $f: \R^n \to \R$ is a convex function that is not necessarily differentiable and $V$ is the feasible set of solutions. Let $\lapprox[f][\bm{x}][\bm{y}]$ be the \mydef{linear approximation} of $f$ at $\bm{y}$, that is
%\begin{align}
$\lapprox[f][\bm{x}][\bm{y}] = f(\bm{y}) + \subgrad^T (\bm{x} - \bm{y})$,
%\enspace ,
%\end{align}
where $\subgrad \in \subdiff[\bm{x}]f(\bm{y})$.
A standard method for solving this problem is the \mydef{mirror descent} \cite{nemirovskij1983problem} update rule is as follows:
\begin{align}
    \bm{x}(t+1) &= \argmin_{\bm{x} \in V} \left\{\lapprox[f][\bm{x}][\x(t)] + \gamma_t \divergence[h][\bm{x}][\bm{x}(t)] \right\} && \text{for }t = 0, 1, 2, \ldots \label{generalized-descent} \\
    \bm{x}(0) &\in \mathbb{R}^{n}\label{generalized-descent2}
\end{align}

\noindent
Here, as above, $\gamma_t > 0$ is the step size at time $t$ and, $\divergence[h][\bm{x}][\bm{x}(t)]$ is the
\mydef{Bregman divergence} of a convex differentiable \mydef{kernel} function $h(\bm{x})$ defined as
%\begin{align}
$\divergence[h][\bm{x}][\bm{y}] = h(\bm{x}) - \lapprox[h][\bm{x}][\bm{y}]$ \cite{bregman1967relaxation}.
%\enspace .
%\end{align}
% Notice that when $h(\bm{x})= \frac{1}{2}||\bm{x}||_2^2$, we have that $\divergence[h][\bm{x}][\bm{y}] = \frac{1}{2} ||\bm{x} - \bm{y}||_2^2$.
% In this case, generalized gradient descent reduces to the projected subgradient method (\Crefrange{subgradient-descent}{subgradient-descent-init}).
% % \begin{align}
% %     \bm{x}(t+1) &= \argmin_{\bm{x} \in V} \left\{\lapprox[f][\bm{x}][\x(t)] + \frac{1}{2\mu_t} ||\bm{x} - \bm{x}(t)||_2^2 \right\} && \text{for }t = 0, 1, 2, \ldots  \\
% %     \bm{x}(0) &\in \mathbb{R}^{n}
% % \end{align}
% If instead 
When the kernel is the scaled weighted entropy $h(\x) = c \sum_{i \in [n]} \left(x_{i} \log(x_{i}) - x_i \right)$, given $c > 0$, then the Bregman divergence reduces to the \mydef{scaled generalized Kullback-Leibler divergence}:
%
% \begin{align}
$    
    \divergence[\mathrm{KL}][\x][\y] = c \sum_{i \in [n]} \left[ x_i \log\left( \frac{x_i}{y_i}\right) - x_i +  y_i \right]
$,
% \end{align}
which, when $V = \R^n_+$, yields the following simplified update rule, where as usual $\subgrad(t) \in \subdiff[\x]f(\x(t))$:
%\amy{i think you need a projection step here as well, in the form of a normalizing constant (i think!).}\deni{I think I know what you are talking about Amy but that specific update rule is when $V = \{\x \mid \sum_{i \in [n]} x_i = 1, x_i \geq 0\}$, i.e., when the set of feasible solutions is the probability simplex. That also makes sense because when the probability simplex is the feasible space, normalizing the prices gets us back to the probability simplex. In this particular case since we do not restrict ourselves to the probability simplex, we do not need the normalization.I am happy to elaborate more on the intuition behind the update step below if you would like to but the main point is that the below update step has no projection step because it never leaves the feasible space of solutions! Additionally, you can find the below update rule used in \href{https://www.sciencedirect.com/science/article/pii/S0899825619300491}{\color{blue}{definition 2.9 in the t\^atonnment beyond gross substitutes paper!}}.}
\begin{align}
    & \forall \good \in \goods & x_{\good}(t+1) &= x_{\good}(t) \exp \left\{\frac{-\subgrad[\good](t)}{\gamma_t} \right\} && \text{for }t = 0, 1, 2, \ldots \label{exp:subgradient-descent} \\
    & & x_{\good}(0) &\in \mathbb{R}_{++}
    \label{exp:subgradient-descent-init}
\end{align}
\Crefrange{exp:subgradient-descent}{exp:subgradient-descent-init} do not include a projection step, because when the initial iterate is within $\R_+^n$, \wine{the update rule guarantees that} subsequent iterates remain within this set.

\csch{A function $f$ is said to be \mydef{$\gamma$-Bregman} with respect to a Bregman divergence with kernel function $h$ if $f(\bm{x}) \leq \lapprox[f][\bm{x}][\bm{y}] + \gamma \divergence[h][\bm{x}][\bm{y}]$.
% \amy{maybe give this condition a name as well: $f(\bm{x}) \leq \lapprox[f][\bm{x}][\bm{y}] + \gamma \divergence[h][\bm{x}][\bm{y}]$.}
\citeauthor{grad-prop-response} \cite{grad-prop-response} showed that if the objective function $f(\bm{x})$ of a convex optimization problem is \mydef{$\gamma$-Bregman} w.r.t.\ to some Bregman divergence $\divergence[h]$, then generalized gradient descent with Bregman divergence $\divergence[h]$ converges to an optimal solution $f(\bm{x}^*)$ at 
%a sublinear rate---specifically, at 
a rate of $O \mleft( \nicefrac{1}{t} \mright)$ \cite{grad-prop-response}.

\if 0
\begin{theorem}{\cite{grad-prop-response}}
\sdeni{Suppose $f$ and $h$ are convex, and for some $\gamma > 0$, we have that $f$ is $\gamma$-Bregman with respect to some Bregman divergence with kernel function $h$.
%\enspace .
%\end{align}
%
If $\bm{x}^*$ is a minimizer of $f$, then for all $t \in \N$, the following holds for generalized gradient descent
%the process described in (\Crefrange{subgradient-descent}{subgradient-descent-init})
with fixed step size:
%\begin{align}
$f(\bm{x}(t)) - f(\bm{x}^*) \leq \nicefrac{\gamma}{t} \left( \divergence[h][\bm{x}^*][\bm{x}(0)] \right)$.
% \amy{check formatting. wanted something inline.}
%\enspace .
%\end{align}
}{}
\end{theorem}
\fi

We require a slightly modified version of this result, shown by \citeauthor{fisher-tatonnement}, where it suffices for the $\gamma$-Bregman property to hold only for consecutive pairs of iterates
%$\x$'s
\cite{fisher-tatonnement}.

\begin{theorem}{\cite{fisher-tatonnement}}\label{theorem-devanur}
Suppose $f$ and $h$ are convex, and for all $t \in \N$ and for some $\gamma > 0$, we have that
%\begin{align}
$f(\bm{x}(t+1)) \leq \lapprox[f][\bm{x}(t+1)][\bm{x}(t)] + \gamma \divergence[h][\bm{x}(t+1)][\bm{x}(t)]$. 
% \amy{just name the condition, since by now, it should have been defined previously.}
%\label{cond1-devanur}$
%\enspace .
%\end{align}
If $\bm{x}^*$ be a minimizer of $f$, then for all t, the following holds for generalized gradient descent
%the process described in
%(\Crefrange{subgradient-descent}{subgradient-descent-init})
with fixed step size:
%\begin{align}
$f(\bm{x}(t)) - f(\bm{x}^*) \leq \nicefrac{\gamma}{t} \divergence[h][\bm{x}^*][\bm{x}(0)]$.
%\enspace .
%\label{devanur-bound}
%\end{align}
\end{theorem}
}

%% file: overview.tex
\subsection{A High-Level Overview of Our Contributions}

%Our paper is centered around \samy{two}{three} main results.

In this paper, we bring consumer theory to bear in the analysis of CCH Fisher markets.
In so doing, we first derive the dual of the Eisenberg-Gale program for arbitrary CCH Fisher markets, generalizing the special cases of linear and Leontief markets, which are already understood~\cite{devanur2008market}.
We then provide a new convex program whose dual also characterizes equilibrium prices in CCH Fisher markets via expenditure functions.
This program is of interest because the subdifferential of the objective function of its dual is equal to the negative excess demand in the market, which implies that mirror descent on this objective is equivalent to solving for equilibrium prices in the associated market via t\^atonnement.
%\csch{Finally, using this equivalence, we extend the sublinear convergence rate of $O(\nicefrac{1}{t})$ known for Leontief Fisher markets to the entire class of CSCH Fisher markets.}
Finally, we conjecture a convergence rate of $O(\nicefrac{(1+E)}{t^2})$ for CCH Fisher markets in which the elasticity of buyer demands is bounded by $E$.\amy{why say $<\infty$? doesn't the word bounded imply $<\infty$?}

%%% OLD, but proofread to double check none of this is needed!

Although the Eisenberg-Gale convex program dates back to 1959, its dual for arbitrary CCH Fisher markets is still not yet well understood.
Our first result is to derive the Eisenberg-Gale program's dual, generalizing the two special cases identified by \citeauthor{cole2016convex} \cite{cole2016convex} for linear and Leontief utilities.
% Studying the Eisenberg-Gale program via the well-understood building blocks of consumer theory, e.g., UMP and EMP, gives us an intuitive understanding of its dual...

\begin{restatable}{theorem}{EGgeneralization}
\label{thm1-overview}\label{EG-generalization}
The dual of the Eisenberg-Gale program for any CCH Fisher market $(\util, \budget)$ is given by:
\begin{align}
    \min_{\price \in \R^\numgoods_+} \sum_{\good \in \goods} \price[\good] + \sum_{\buyer \in \buyers} \left[ \budget[\buyer] \log{ \left( \indirectutil[\buyer](\price, \budget[\buyer]) \right)} - \budget[\buyer] \right] \label{convex-program1}
\end{align}
\end{restatable}

We then propose a new convex program whose dual characterizes the equilibrium prices of CCH Fisher markets via expenditure functions.
We note that the optimal value of this convex program differs from the optimal value of the Eisenberg-Gale program by a constant factor.

\begin{restatable}{theorem}{newconvex}
\label{thm2-overview}\label{new-convex}
% \label{new-convex}
The optimal solution $(\allocation^*, \price^*)$ to the primal and dual of the following convex programs corresponds to equilibrium allocations and prices, respectively, of the CCH Fisher market $(\util, \budget)$:
\begin{equation*}
\begin{aligned}[c]
    &\text{\textbf{Primal}} \\  
    & \max_{\allocation \in \R_+^{\numbuyers \times \numgoods}}  \sum_{\buyer \in \buyers} 
    % \left[ 
    \budget[\buyer] \log{\left(\frac{\util[\buyer]\left(\allocation[\buyer]\right)}{\budget[\buyer]}\right)} 
    % + \budget[\buyer] \right] 
    \\
    & \text{subject to}  \quad \forall \good \in \goods
    % \in \goods
    , \sum_{\buyer \in \buyers} \allocation[\buyer][\good]  \leq 1 \\
    % && \allocation[\buyer][\good] &\geq 0 &&\forall \buyer \in \buyers, \good \in \goods 
\end{aligned}
\quad \vline \quad
\begin{aligned}[c]
    &\text{\textbf{Dual}} \\  &\min_{\price \in \R^\numgoods_+} \sum_{\good \in \goods} \price[\good] - \sum_{\buyer \in \buyers} \budget[\buyer] \log{ \left( \partial_{\goalutil[\buyer]}{\expend[\buyer] (\price, \goalutil[\buyer])} \right)}
\end{aligned}
\end{equation*}
\end{restatable}

This convex program formulation for CCH Fisher markets is of particular interest because its subdifferential equals the negative excess demand in the market.
As a result, solving this program via (sub)gradient descent is equivalent to solving the market via t\^atonnement. 

\begin{restatable}{theorem}{excessdemand}
\label{thm3-overview}\label{excess-demand}
The subdifferential of the objective function of the dual of the program given in \Cref{thm2-overview} for a CCH Fisher market $(\util, \budget)$ at any price $\price$ is equal to the negative excess demand in $(\util, \budget)$ at price $\price$:
\begin{align}   
    \subdiff[\price] \left( \sum_{\good \in \goods} \price[\good] - \sum_{\buyer \in \buyers} \budget[\buyer] \log{ \partial_{\goalutil[\buyer]}{\expend[\buyer] (\price, \goalutil[\buyer])}} \right) %= \bm{1}_{\numgoods} - \sum_{\buyer \in \buyers} \marshallian[\buyer](\price, \budget[\buyer])  
    = - \excess(\price)
\end{align}
\end{restatable}

To prove \Cref{thm3-overview}, we make use of standard consumer theory, specifically the duality structure between UMP and EMP, as well as a generalized version of Shepherd's lemma \cite{shephard, generalized-shephard}.
We also provide a new, simpler proof of this generalization of Shepherd's lemma via Danskin's theorem \cite{danskin1966thm}.

Finally, we conduct an experimental investgation of the convergence of the t\^atonnement process defined by the mirror descent rule with KL-divergence and fixed step sizes
%given in \Crefrange{generalized-descent}{generalized-descent2}
in CCH Fisher markets.
This particular process was previously studied by \citeauthor{fisher-tatonnement} \cite{fisher-tatonnement} in Leontief Fisher markets.
They showed a worst-case \emph{lower\/} bound of $\Omega(\nicefrac{1}{t^2})$ to complement an $O(\nicefrac{1}{t})$ worst-case upper bound.
These results suggest a possible convergence rate of $O(\nicefrac{1}{t^2})$ or $O(\nicefrac{1}{t})$ for entropic t\^atonnement in a class of Fisher markets that includes Leontief Fisher markets.
Our experimental results support the conjecture that a worst-case convergence rate of $O(\nicefrac{1}{t^2})$ might hold, not only in Leontief and CES Fisher markets, but in CCH Fisher markets where buyers' elasticity of demand is bounded by $E$.\amy{why say $<\infty$? doesn't the word bounded imply $<\infty$?}\deni{Just makes it clear I guess, I generally get confused when authors say bounded because }

\if 0
\sdeni{}{We conjecture that a convergence rate of \amy{conjecture!} $O(\nicefrac{1}{t^2})$ can be proven for this t\^atonnement rule in CCH Fisher markets with bounded elasticity of demand based on experimental evidence. This also makes sense since the larger the elasticity of demand, the more responsive the demand is to increases in prices and as a result, the larger the step size should be so that convergence can be ensured. That is, if elasticity of demand is bounded, then the t\^atonnement price updates can be arbitrarily large causing it to cycle. If later work proves such a result, this would not only provide convergence guarantees for t\^atonnmenet in the most general setting possible but it would also improve the known sublinear convergence rate of \amy{conjecture!} $O(\nicefrac{1}{t})$ of this t\^atonnement rule for Leontief Fisher markets.} 
\fi

\csch{Our result extends the sublinear convergence rate of \amy{conjecture!} $O(\nicefrac{1}{t})$ they obtained for Leontief Fisher markets to the entire class of CSCH Fisher markets, which includes classes of markets for which convergence results were not previously known, such as mixed and nested CES markets \cite{cheung2018dynamics, jain2006equilibria}.
We note that it is our novel characterization of equilibrium prices via expenditure functions that allows us to generalize \citeauthor{fisher-tatonnement}'s \cite{fisher-tatonnement} convergence result to all CSCH Fisher markets.

\begin{theorem}
\label{thm4-overview}
Let $f(\price) = \sum_{\good \in \goods} \price[\good] - \sum_{\buyer \in \buyers} \budget[\buyer] \log{ \left( \partderiv[{\expend[\buyer] (\price, \goalutil[\buyer])}][{\goalutil[\buyer]}] \right)}$. The following holds for the t\^atonnement process
%given in \Crefrange{tatonnement}{tatonnement2}
when run on a CSCH Fisher market $(\util, \budget)$:
%\begin{align}
$f(\price(t)) - f(\price^*) \leq \nicefrac{1}{t} \left( \gamma \divergence[\mathrm{h}][\price^*][\price(0)] \right)$,
%\enspace ,
%\end{align}
where $\gamma = 5 \max_{\substack{\good \in \goods\\ t \in \N}} \left\{\sum_{\buyer \in \buyers}\marshallian[\buyer][\good]\left(\price(t), \budget[\buyer]\right) \right\}$.
% 10 \cdot \max_{\good \in \goods} \left\{ \left[ \frac{\sum_{\buyer \in \buyers} \budget[\buyer]}{\price[\good]^0} \max \left\{\max_{k \in \goods}\demand[k]^0,  \frac{\max_{l \in \goods}\demand[l]^0}{\demand[\good]^0} \right\} + \max_{k \in \goods} \frac{\demand[\good]^0}{  \demand[k]^0} \right] \right\}$ and $\divergence[\mathrm{h}] = 6 \cdot \divergence[\mathrm{KL}][][]
\end{theorem}

\noindent
We provide an upper bound on $\gamma$ for mixed CCH Fisher markets
% \amy{why not call this $\gamma$? or $\gamma/5$, if you prefer?} $\max_{\substack{\good \in \goods\\ t \in \N}} \left\{\sum_{\buyer \in \buyers}\marshallian[\buyer][\good]\left(\price(t), \budget[\buyer]\right) \right\}$ 
in \Cref{sec:equiv}, so that the step size can be set when the process is initialized.
For other markets, we provide a naive upper bound when the number of iterations for which t\^atonnement will be run in known in advance, which can then be combined with the doubling trick 
% \cite{hazan2014doubling}\amy{why this particular reference?}\deni{Seemed like the thing to cite but I am clearly wrong!}\amy{the doubling trick is VERY old. i am not sure what to cite actually. but maybe nothing?}\deni{Ok, i guess it is common knowledge then!}
so that the convergence bound holds for t\^atonnement run without a known time horizon.
We note that our convergence bounds can be improved by a constant factor in future work that derives tighter bounds on the demand, since 
%the learning rate 
$\gamma$ is a function of the demand.}

% \sdeni{At a high level, our approach follows the approach taken by  \citeauthor{fisher-tatonnement} \cite{fisher-tatonnement}. First, we prove a lemma that states assumptions under which the condition of \Cref{theorem-devanur} holds when $f(\price) = \sum_{\good \in \goods} \price[\good] - \sum_{\buyer \in \buyers} \budget[\buyer] \log{ \left( \partderiv[{\expend[\buyer] (\price, \goalutil[\buyer])}][{\goalutil[\buyer]}] \right)}$. For these assumptions to hold, we require that the learning rate $\gamma$ be five times the maximum demand for any good throughout the t\^atonnment process. In order to set $\gamma$ at the beginning of t\^atonnement, we upper bound this maximum demand. Once this bound is derived, we apply \Cref{theorem-devanur} to obtain the convergence rate of $O(\nicefrac{1}{t})$.}{}

%% file: newconvex.tex
\section{A New Convex Program for CCH Fisher Markets}\label{sec:program}

In this section, we provide an alternative convex program to the Eisenberg-Gale program, which also characterizes the equilibria of CCH Fisher markets.
Of note, our program characterizes equilibrium prices via expenditure functions.
For CCH Fisher markets, the Eisenberg-Gale program's primal allows us to calculate the equilibrium allocations, while its dual yields the corresponding equilibrium prices \cite{chen2007note}.
%\deni{Even if the dual was not known, we know the dual gives equilibrium prices because the slack variables corresponding to the allocation constraints are interpreted as prices in the proof of optimality of the EG program for Fisher markets equilibria. Since the slack variables become the variables in the dual, even if we do not know the dual, we know that the prices are the decision variables in the dual!}
\citeauthor{cole2016convex} \cite{cole2016convex} provide dual formulations of the Eisenberg-Gale program for linear and Leontief utilities \cite{cole2016convex}, and in unpublished work, \citeauthor{cole2019balancing} \cite{cole2019balancing} present a generalization of the Eisenberg-Gale dual for arbitrary CCH utility functions. However, as we show in \wineusesparingly{\Cref{dual-diff-cole} (\Cref{sec:Proofs3})}{the full version}, the optimal value of the objective of the Eisenberg-Gale program's primal differs from the optimal value of the dual provided by \citeauthor{cole2019balancing} \cite{cole2019balancing} by a constant factor, despite their dual characterizing equilibrium prices accurately.
%\amy{do you want to double check that they haven't updated their paper recently?}\deni{Not updated!}
Hence, their dual is technically not the dual of the Eisenberg-Gale program for which strong duality holds.
%\amy{let's just please be very careful not to be insulting.}
The proof of the following theorem stating the Eisenberg-Gale program's dual can be found in \wineusesparingly{\Cref{sec:Proofs3}}{the full version}.

\EGgeneralization*

Before presenting our program, we present several preliminary lemmas.
All omitted proofs can be found in \wineusesparingly{\Cref{sec:Proofs4}}{the full version}.

The next lemma establishes an important property of the indirect utility and expenditure functions in CCH Fisher markets that we heavily exploit in this work, namely that the derivative of the indirect utility function with respect to $\budget[\buyer]$---the bang-per-buck---is constant across all budget levels.
Likewise, the derivative of the expenditure function with respect to $\goalutil[\buyer]$---the buck-per-bang---is constant across all utility levels.
In other words, both functions effectively depend only on prices.
Not only are the bang-per-buck and the buck-per-bang constant, they equal $\indirectutil[\buyer](\price, 1)$ and $\expend[\buyer](\price, 1)$, respectively, namely their values at exactly one unit of budget and one unit of (indirect) utility.

%\amy{what do you think of using the terms cost and benefit?}\deni{I prefer cost and utility, I've seen utility of one unit of budget used in other works in the past.}\amy{i prefer benefit, since utility is overloaded. it is like cost, instead of budget or price. but on the flip side, we have utility, and values, and then utility again.}

An important consequence of this lemma is that, by picking prices that maximize a buyer's bang-per-buck, we not only maximize their bang-per-buck at all budget levels, but we further maximize their total indirect utility, given their \emph{known\/} budget.
In particular, given prices $\price^*$ that maximize a buyer's bang-per-buck at budget level 1, we can easily calculate the buyer's total (indirect) utility at budget $\budget[\buyer]$ by simply multiplying their bang-per-buck by $\budget[\buyer]$: i.e., $\indirectutil[\buyer](\price^*, \budget[\buyer]) = \budget[\buyer] \indirectutil[\buyer](\price^*, 1)$.
Here, we see \wine{quite explicitly} the homogeneity assumption at work.

Analogously, by picking prices that maximize a buyer's buck-per-bang, we not only maximize their buck-per-bang at all utility levels, but we further maximize the buyer's total expenditure, given their \emph{unknown\/} optimal utility level.
In particular, given prices $\price^*$ that minimize a buyer's buck-per-bang at utility level 1,
we can easily calculate the buyer's total expenditure at utility level $\goalutil[\buyer]$ by simply multiplying their buck-per-bang by $\goalutil[\buyer]$: i.e., $\expend[\buyer](\price^*, \goalutil[\buyer]) = \goalutil[\buyer] \expend[\buyer](\price^*, 1)$.
%Again, we see quite explicitly the homogeneity assumption at work.
Thus, solving for optimal prices at any budget level, or analogously at any utility level, requires only a single optimization, in which we solve for optimal prices at budget level, or utility level, 1.

\begin{restatable}{lemma}{deriveexpend}
\label{derive-expend}
\label{deriv-indirect-util}
If $\util[\buyer]$ is continuous and homogeneous of degree 1, then $\indirectutil[\buyer](\price, \budget[\buyer])$ and
$\expend[\buyer](\price, \goalutil[\buyer])$ are differentiable in $\budget[\buyer]$ and $\goalutil[\buyer]$, resp.
Further, $\subdiff[{\budget[\buyer]}] \indirectutil[\buyer](\price, \budget[\buyer]) = \left\{\indirectutil[\buyer](\price, 1) \right\}$ and
$\subdiff[{\goalutil[\buyer]}] \expend[\buyer](\price, \goalutil[\buyer]) = \left\{\expend[\buyer](\price, 1) \right\}$.
\end{restatable}

%It turns out that the Hicksian demand and the expenditure function are not only related through the expenditure minimization problem (EMP).

The next lemma provides further insight into why CCH Fisher markets are easier to solve than non-CCH Fisher markets.
%For CCH utility functions,
The lemma states that the bang-per-buck, i.e., the marginal utility
%\amy{benefit?}
of an additional unit of budget, is equal to the inverse of its buck-per-bang, i.e., the marginal cost of an additional unit of utility.
Consequently, by setting prices so as to minimize the buck-per-bang of buyers, we can also maximize their bang-per-buck.
%, which by \Cref{deriv-indirect-util} also maximizes each buyer's utility constrained by their budget.
Since the buck-per-bang is a function of prices only, and not of prices and allocations together, this lemma effectively decouples the calculation of equilibrium prices from the calculation of equilibrium allocations, which
%The fact that the optimization of prices can be decoupled from allocation computation, as implied by this lemma, 
greatly simplifies the problem of computing equilibria in CCH Fisher markets.

% \deni{Combine Lemma 3 and corollary 1 and ensure to the fact the expenditure}

% \begin{lemma}
% \label{general-indirect-expend-rel}
% If buyer $\buyer$'s utility function $\util[\buyer]$ is CCH,
% %Let $\expend[\buyer](\price, \goalutil[\buyer])$ be the expenditure function of buyer $\buyer$ and $\indirectutil[\buyer](\price, \budget[\buyer])$ be the indirect utility function of buyer $\buyer$. 
% then $\subdiff[{\goalutil[\buyer]}] \expend[\buyer](\price, \goalutil[\buyer]) \cdot \subdiff[{\budget[\buyer]}] \indirectutil[\buyer](\price, \budget[\buyer]) = \{1\}$.
% \end{lemma}

\begin{restatable}{corollary}{inverseexpend}
\label{inverse-expend}
If buyer $\buyer$'s utility function $\util[\buyer]$ is CCH,
%Let $\expend[\buyer](\price, \goalutil[\buyer])$ be the expenditure function of buyer $\buyer$ and $\indirectutil[\buyer](\price, \budget[\buyer])$ be the indirect utility function of buyer $\buyer$.
then
\begin{align}
    \frac{1}{\expend[\buyer](\price, 1)} = \frac{1}{\partderiv[{\expend[\buyer](\price, \goalutil[\buyer])}][{\goalutil[\buyer]}]} = \partderiv[{\indirectutil[\buyer](\price, \budget[\buyer])}][{\budget[\buyer]}] = \indirectutil[\buyer](\price, 1)
    \enspace .
\end{align}
%\amy{is $\partderiv[{\indirectutil[\buyer](\price, \budget[\buyer])}][{\budget[\buyer]}]$ well defined? why isn't this a set? likewise for the expenditure function?}\deni{$\partderiv[{\indirectutil[\buyer](\price, \budget[\buyer])}][{\budget[\buyer]}] = \indirectutil[\buyer](\price, 1)$ and $\indirectutil[\buyer](\price, 1)$ is the minimum (not minimizer) of an optimization problem, hence, it can only be unique, i.e., it is not set valued. Same for the expenditure function.}
\end{restatable}

% The fact that for CCH utilities the Marshallian consumer surplus obtained by buyer $\buyer$ due to obtaining a good can be written as $ \budget[\buyer] \log \left(\partderiv[{\expend[\buyer](\price, \goalutil[\buyer])}][{\goalutil[\buyer]}] \right)$ follows easily from an extension of the definition of Marshallian consumer surplus to markets with multiple good and the fundamental theorem of calculus. Let $\mathrm{MCS}_{\buyer}: \R^{\numgoods} \times \R^{\numgoods} \to \R$ be the Marshallian consumer surplus due to a change in prices from $\price^0$ to $\price^1$ defined as $\mathrm{MCS}_{\buyer}(\price^0, \price^1) = \int_{\price[1]^0}^{\price[1]^1} \hdots \int_{\price[\numgoods]^0}^{\price[\numgoods]^1} \marshallian[\buyer](\price, \budget[\buyer]) d\price[\numgoods] \hdots d\price[1]$. Then, for CCH utilities the consumer surplus from $\price^0 = \bm{0}$, i.e., before trade, to $\price^1 = \price$, i.e., after trade, can be written as:

% \begin{align}
%     \mathrm{MCS}_{\buyer}(\bm{0}, \price) =  \int_{\price[1]^0}^{\price[1]^1} \hdots \int_{\price[\numgoods]^0}^{\price[\numgoods]^1} \marshallian[\buyer](\price, \budget[\buyer]) d\price[\numgoods] \hdots d\price[1]
% \end{align}

We can now present our characterization of the dual of the Eisenberg-Gale program via expenditure functions.
While \citeauthor{devanur2016new} \cite{devanur2016new} provided a method to construct a similar program to that given in \Cref{new-convex} for specific utility functions, their method does not apply to arbitrary CCH utility functions. The proof of this theorem can be found in \wineusesparingly{\Cref{sec:Proofs4}}{the full version}.

\newconvex*

\if 0
\begin{theorem}
The optimal solution $(\allocation^*, \price^*)$ to the primal and dual of the following convex programs corresponds to equilibrium allocations and prices, respectively, of the CCH Fisher market $(\util, \budget)$:
%
\begin{equation*}
\begin{aligned}[c]
    &\text{\textbf{Primal}} \\  
    &\max_{\allocation \in \R_+^{\numbuyers \times \numgoods}}  \sum_{\buyer \in \buyers} \left[ \budget[\buyer] \log{\util[\buyer]\left(\frac{\allocation[\buyer]}{\budget[\buyer]}\right)} + \budget[\buyer] \right] \\
    & \text{subject to} \quad \forall \good \in \goods, \ \ \sum_{\buyer \in \buyers} \allocation[\buyer][\good]  \leq 1 \\
    % && \allocation[\buyer][\good] &\geq 0 &&\forall \buyer \in \buyers, \good \in \goods 
\end{aligned}
\quad \vline \quad
\begin{aligned}[c]
    &\text{\textbf{Dual}} \\  &\min_{\price \in \R^\numgoods_+} \sum_{\good \in \goods} \price[\good] - \sum_{\buyer \in \buyers} \budget[\buyer] \log{ \left( \partial_{\goalutil[\buyer]}{\expend[\buyer] (\price, \goalutil[\buyer])} \right)}
\end{aligned}
\end{equation*}

% \amy{wait, is this theorem really correct? can you really just drop this term: $\sum_{\buyer \in \buyers} \budget[\buyer]\log \budget[\buyer]$. the eqm prices are the same, but i don't think it is accurate to call this the actual dual of the EG program, since its objective value does not equal the objective value of the dual.}
\end{theorem}
\fi

Our new convex program for CCH Fisher markets, which characterizes equilibrium expenditure functions, makes plain the duality structure between utility functions and expenditure functions that is used to compute ``shadow'' prices for allocations.
In particular, $\expend[\buyer](\price, \goalutil[\buyer])$ is the Fenchel conjugate of the indicator function $\chi_{\{\x: \util[\buyer](\allocation[\buyer]) \geq \goalutil[\buyer]\}}$, 
%i.e., $\expend[\buyer](\price, \goalutil[\buyer]) = \chi^*_{\{\x: \util[\buyer](\allocation[\buyer]) \geq \goalutil[\buyer]\}}$, where $\chi^*$ is the Fenchel conjugate of $\chi$~\cite{duality-economics}, 
meaning the utility levels and expenditures are dual (in a colloquial sense) to one another.
Therefore, equilibrium utility levels can be determined from equilibrium expenditures, and vice-versa, which implies that allocations and prices can likewise be derived from one another through this duality structure.%
\footnote{A more in-depth analysis of this duality structure can be found in \citeauthor{duality-economics} \cite{duality-economics}.}
% \amy{don't use citations as nouns}

%% file: equivalence.tex
\section{Equivalence of Mirror Descent and T\^atonnement}
\label{sec:equiv}

\citeauthor{fisher-tatonnement} \cite{fisher-tatonnement} have shown via the Lagrangian of the Eisenberg-Gale program, i.e., without constructing the precise dual, that the subdifferential of the dual of the Eisenberg-Gale program %\amy{it's a bit weird that they were able to show this, since the dual was not known. maybe add a footnote that their dual also differed from the actual EG dual by a constant factor?}\deni{But they did not know the dual! they just used a proof by cases using the Lagrangian.}
is equal to the negative excess demand in the associated market, which implies that mirror descent equivalent to a subset of t\^atonnement rules.
In this section, we use a generalization of Shephard's lemma to prove that the subdifferential of the dual of our new convex program is equal to the negative excess demand in the associated market.
Our proof also applies to the dual of the Eisenberg-Gale program, since the two duals differ only by a constant factor.

\wine{Shephard's lemma tells us that the rate of change in expenditure with respect to prices, evaluated at prices $\price$ and utility level $\goalutil[\buyer]$, is equal to the Hicksian demand at prices $\price$ and utility level $\goalutil[\buyer]$.
%\deni{This is wrong, previous lemma is a derivative with respect to the utility level, this is a derivative with respect to prices.} \deni{That said, what you are saying is correct for \textbf{strictly concave utility functions} but bit because of Shepherd's lemma; it is just a consequence of homogeneity.}
%
%\amy{the above comments are wrt total expenditure across all goods. the next sentence moves from subdifferentials to partials.}
%
%The proof of this lemma (in the case of strictly concave utility functions) uses the envelope theorem to show that
Alternatively, the partial derivative of the expenditure function with respect to the price $\price[\good]$ of good $\good$ at utility level $\goalutil[\buyer]$ is simply the share of the total expenditure allocated to $\good$ divided by the price of $\good$, which is exactly the Hicksian demand for $\good$ at utility level $\goalutil[\buyer]$.}{}

While Shephard's lemma is applicable to utility functions with singleton-valued Hicksian demand (i.e., strictly concave utility functions), we require a generalization of Shephard's lemma that applies to utility functions that are not strictly concave and that could have set-valued Hicksian demand.
%To the best of our knowledge, 
An early proof of this generalized lemma was given by \citeauthor{generalized-shephard} in a discussion paper \cite{generalized-shephard};
a more modern perspective can be found in a recent survey by \citeauthor{duality-economics} \cite{duality-economics}.
For completeness, we also provide a new, simple proof of this result via Danskin's theorem (for subdifferentials) \cite{danskin1966thm} \wineusesparingly{in \Cref{sec:Proofs5}}{the full version}.

\begin{restatable}{lemma}{shepherd}
\textbf{Shephard's lemma, generalized for set-valued Hicksian demand \cite{duality-economics, shephard, generalized-shephard}}
\label{shepherd}
Let $\expend[\buyer](\price, \goalutil[\buyer])$ be the expenditure function of buyer $\buyer$ and $\hicksian[\buyer](\price, \goalutil[\buyer])$ be the Hicksian demand set of buyer $\buyer$.
The subdifferential $\subdiff[{\price}] \expend[\buyer](\price, \goalutil[\buyer])$ is the Hicksian demand at prices $\price$ and utility level $\goalutil[\buyer]$, i.e., $\subdiff[\price] \expend[\buyer](\price, \goalutil[\buyer]) = \hicksian[\buyer](\price, \goalutil[\buyer])$.
%That is, $\bm{\chi}_\buyer \in \subdiff[{\price}] \expend[\buyer](\price, \goalutil[\buyer])$ iff $\bm{\chi}_{\buyer} \in \hicksian[\buyer](\price, \goalutil[\buyer])$.
%\amy{i like this notation, without the fraction, better. why don't we use it everywhere -- instead of fractions.}\deni{I avoid using it in cases where the object is supposed to not be a set, I guess we could abuse notation but we should say it somehwere if you prefer that.}
\end{restatable}

The next lemma plays an essential role in the proof that the subdifferential of the dual of our convex program is equal to the negative excess demand.
% \amy{see the two remarks below this lemma. it seems to play an essential role in much more than just what we say in this sentence! can we expand on its importance?}\deni{I like the below two lemmas but I am not sure adding more here will be providing much value?}
Just as Shephard's Lemma related the expenditure function to Hicksian demand via (sub)gradients, this lemma relates the expenditure function to Marshallian demand via (sub)gradients.
One way to understand this relationship is in terms of \mydef{Marshallian consumer surplus},
%\amy{except that the reader has no intuition for consumer surplus :(} 
the area under the Marshallian demand curve, i.e., the integral of Marshallian demand with respect to prices.%
\footnote{We note that the definition of Marshallian consumer surplus for multiple goods requires great care and falls outside the scope of this paper. More information on consumer surplus can be found in \citeauthor{levin-notes} \cite{levin-notes}, and \citeauthor{vives1987marshallian} \cite{vives1987marshallian}.}
Specifically, by applying the fundamental theorem of calculus to the left-hand side of \Cref{lemma-deriv-marshallian}, we see that the Marshallian consumer surplus equals $\budget[\buyer] \log \left(\partderiv[{\expend[\buyer](\price, \goalutil[\buyer])}][{\goalutil[\buyer]}] \right)$.
The key takeaway is thus that any objective function we might seek to optimize that includes a buyer's Marshallian consumer surplus is thus optimizing their Marshallian demand, so that optimizing this objective yields a utility-maximizing allocation for the buyer, constrained by their budget.

\begin{restatable}{lemma}{lemmaderivmarshallian}
\label{lemma-deriv-marshallian}
If buyer $\buyer$'s utility function $\util[\buyer]$ is CCH,
%Let $\expend[\buyer](\price, \goalutil[\buyer])$ be the expenditure function of buyer $\buyer$ and let $\marshallian[\buyer](\price, \budget[\buyer])$ be the Marshallian demand of buyer $\buyer$.
then
%\begin{align}
$\subdiff[\price] \left( 
    \budget[\buyer] \log \left(\partderiv[{\expend[\buyer](\price, \goalutil[\buyer])}][{\goalutil[\buyer]}]  \right) \right) = 
    \marshallian[\buyer](\price, \budget[\buyer])$.
%   \enspace .
%\end{align}
\end{restatable}

% \amy{seems like Lemma 4 should move as well, since it is not needed above!}

\begin{remark}
\Cref{lemma-deriv-marshallian} makes the dual of our convex program easy to interpret,
%the objective function of 
and thus sheds light on the dual of the Eisenberg-Gale program.
%%%
%, which had previously eluded us.
\if 0
Consider the sub-expression $- \sum_{\buyer \in \buyers} \budget[\buyer] \log{ \left(\partderiv[{\expend[\buyer](\price, \goalutil[\buyer])}][{\goalutil[\buyer]}] \right)}$.
Since the utility functions are assumed to be CCH, maximizing the budget-weighted logarithm of the marginal cost of utility is equivalent to maximizing the Marshallian demand, or the total of the buyers' utilities constrained by their budgets (\Cref{lemma-deriv-marshallian}).
Second, consider the sub-expression $\sum_{\good \in \goods} \price[\good]$.
Minimizing prices is equivalent to calibrating prices such that the market clears.
Hence, the objective function can be seen as the efforts of a fictional auctioneer who is trying to maximize the utility of each buyer constrained by their budget, while controlling prices so that the market clears. \amy{i am missing the market clearing intuition. i like the next intuition better, about minimizing distance. maybe i need to see some examples. is there a way to interpret the VALUE of our convex program?}
\fi
%
Specifically, we can interpret the dual as specifying prices that minimize the distance between the sellers' surplus and the buyers' \wineusesparingly{}{Marshallian} surplus.
\wine{The left hand term is simply the sellers' surplus, and by \Cref{lemma-deriv-marshallian}, the right hand term can be seen as the buyers' total Marshallian surplus.}
\deni{in general the surplus of a seller is their profit}
\end{remark}

%\amy{when you find the dual, you abstract out the computation of the primal---the allocation. and likewise, for the primal. so there is a separation of concerns. need to argue why the two claims above about optimizing buck per bang or bang per buck lead to separation of concerns.}

\begin{remark}
The lemmas we have proven in this section and the last provide a possible explanation as to why no primal-dual type
%\amy{remind me why we added the ``primal-dual'' qualifier. do we need it?}\deni{Yes, we need it. because a convex program to solve for prices on their own might exists, but this might not mean that its dual will gives equilibrium allocations necessarily}
convex program is known that solves Fisher markets when buyers have \emph{non-homogeneous\/} utility functions, in which the primal describes optimal allocations while the dual describes equilibrium prices.
By the homogeneity assumption, a CCH buyer can increase their utility level (resp. decrease their spending) by $c\%$ by increasing their budget (resp. decreasing their desired utility level) by $c\%$\wineusesparingly{ (\Cref{homo-indirect-util}); see \Cref{sec:Proofs3} }{}.
This observation implies that the marginal expense of additional utility, i.e., ``bang-per-buck'', and the marginal utility of additional budget, i.e., ``buck-per-bang'', are constant (\Cref{deriv-indirect-util}).
Additionally,  optimizing prices to maximize buyers'  ``bang-per-buck'' is equivalent to optimizing prices to minimize their ``buck-per-bang'' (\Cref{inverse-expend}).
Further, optimizing prices to minimize their ``buck-per-bang'' is equivalent to maximizing their utilities constrained by their budgets (\Cref{lemma-deriv-marshallian}).
Thus, the equilibrium prices computed by the dual of our program, which optimize the buyers' buck-per-bang, simultaneously optimize their utilities constrained by their budgets.
In particular, equilibrium prices can be computed without reference to equilibrium allocations (\Cref{inverse-expend} + \Cref{lemma-deriv-marshallian}).
In other words, assuming homogeneity, the computation of the equilibrium allocations and prices can be isolated into separate primal and dual problems.
\end{remark}

Next, we show that the subdifferential of the dual of our convex program is equal to the negative excess demand in the associated market.

\excessdemand*
\if 0
\begin{theorem}
\label{excess-demand}
The subdifferential of the dual of the program given in \Cref{new-convex} for a CCH Fisher market $(\util, \budget)$ at any price $\price$ is equal to the negative excess demand in $(\util, \budget)$ at price $\price$:
%
\begin{align}   
    \subdiff[\price] \left( \sum_{\good \in \goods} \price[\good] - \sum_{\buyer \in \buyers} \budget[\buyer] \log{ \partial_{\goalutil[\buyer]}{\expend[\buyer] (\price, \goalutil[\buyer])}} \right) 
    % = \bm{1}_{\numgoods} - \sum_{\buyer \in \buyers} \marshallian[\buyer](\price, \budget[\buyer]) 
    = - \excess(\price)
\end{align}
\end{theorem}
\fi

%This theorem implies that the dual of our convex program is a convex potential function for CCH Fisher markets, a fact previously proven by \citeauthor{fisher-tatonnement} for the dual of the Eisenberg-Gale program for which no explicit dual formulation was known.

\citeauthor{fisher-tatonnement} \cite{fisher-tatonnement} define a class of markets called \mydef{convex potential function (CPF)} markets.
A market is a CPF market, if there exists a convex potential function $\varphi$ such that $\subdiff[\price] \varphi(\price) = - \excess(\price)$.
They then prove that Fisher markets are CPF markets by showing, through the Lagrangian of the Eisenberg-Gale program, that its dual is a convex potential function \cite{fisher-tatonnement}.
Likewise, \Cref{excess-demand} implies the following:

\begin{corollary}
All CCH Fisher markets are CPF markets.
\end{corollary}

\begin{proof}
A convex potential function $\phi: \mathbb{R}^{\numgoods} \to \mathbb{R}$ for any CCH Fisher market $(\util, \budget)$ is given by:
\begin{align}
    \potential (\price) = \sum_{\good \in \goods} \price[\good] - \sum_{\buyer \in \buyers} \budget[\buyer] \log{ \left( \partderiv[{\expend[\buyer] (\price, \goalutil[\buyer])}][{\goalutil[\buyer]}] \right)}\label{potential-func}
\end{align}
\end{proof}

Fix a kernel function $h$ for the Bregman divergence $\divergence[h]$.
If the mirror descent procedure given in  \Crefrange{generalized-descent}{generalized-descent2} is run on \Cref{potential-func} (i.e., choose $f = \potential$), it is then equivalent to the t\^atonnement process for some monotonic function of the excess demand \cite{fisher-tatonnement}\wineusesparingly{:}{.}
\wine{\begin{align}
    \price(t+1) = \min_{\price \in \mathbb{R}^\numgoods} \left\{ \price(t) +  \subgrad(t)\left(\price - \price(t) \right) + \gamma_t \divergence[\mathrm{h}][\price][\price(t)] \right\}  && \text{for } t = 0, 1, 2, \hdots\\
    \subgrad(t) \in \subdiff[\price] \potential(\price(t))\\
    \price(0) \in \mathbb{R}^\numgoods_+
\end{align}
}{}

Thus, by varying the kernel function $h$ of the Bregman divergence we can obtain different t\^atonnement rules.
For instance, if $h = \frac{1}{2}||x||^2_2$, the mirror descent process reduces to the classic t\^atonnement rule given by $G(\x) = \gamma_t \x$, for $\gamma_t > 0$ and for all $t \in \N$, in \Crefrange{tatonnement}{tatonnement2}.

\csch{Using this equivalence between \samy{generalized gradient}{mirror} descent and t\^atonnement, we can pick a particular kernel function $h$, and then potentially use \Cref{theorem-devanur} to establish convergence rates for t\^atonnement.}

%% file: tatonnement.tex
\section{Convergence of Discrete T\^atonnement}
\label{sec:convergence}

In this section, we conduct an experimental investigation%
\footnote{Our code can be found on \coderepo.}
of the rate of convergence of \mydef{entropic t\^atonnement}, which corresponds to the t\^atonnement process given by mirror descent
% with kernel function $h(\x) = \sum_{i \in [n]} \left(x_{i} \log(x_{i}) - x_i \right)$, which corresponds to \samy{generalized gradient}{mirror} descent 
with the scaled generalized Kullback-Leibler (KL) divergence, specifically $6 \divergence[\mathrm{KL}][\p][\q]$, as the Bregman divergence, and a fixed step size $\gamma$.
This particular update rule, which reduces to \Crefrange{exp:subgradient-descent}{exp:subgradient-descent-init}, has been the focus of previous work \cite{fisher-tatonnement}.
Interest in this update rule stems from the fact that prices can never reach 0, which ensures that demands, and as a consequence, excess demands, are bounded throughout the t\^atonnement process. 
This is because the demand for any good $\good$ is always upper bounded by $\frac{\sum_{\buyer \in \buyers} \budget[\buyer]}{\price[\good]}$. Before presenting experimental results for entropic t\^atonnement, we note that the process is not guaranteed to converge in all CCH Fisher markets. 
It does not converge, for example, in linear Fisher markets:
\wineusesparingly{}{An example of such a market can be found in the full version.} \wine{The following market,%
\footnote{We thank an anonymous reviewer of an earlier version of this paper for providing this counterexample.} which assumes buyers with linear utility functions, and is a slightly modified version of an example provided by \citeauthor{cole2019balancing} \cite{cole2019balancing}.
\begin{example}
Consider a linear Fisher market with two goods and one buyer with utility function $\util[ ](x_1, x_2) = x_1 + x_2 $ and budget $b = 1$.
Assume initial prices of $p_1(0) = 1$ and $p_2(0) = e^{1/\gamma}$, for any $\gamma > 0$.
In the first iteration, the demand for the first good is 1, while demand for the second is 0.
Therefore, the prices during the second iteration are $p_1((1)) = (1)e^{1/\gamma} = e^{1/\gamma} $ and $ p_2((1)) = e^{1/\gamma} e^{-1/\gamma} = 1$.
As the prices cycle, so too does t\^atonnement.
\label{ex:linear}
\end{example}}
%
\if 0
%%% READ LATER !!!
\amy{but is there a lower bound on prices? because if there isn't, it still feels like demand is unbounded?}\deni{The price of a good can decrease at most such that $\price^{t+1} = \price^t e^{\nicefrac{-1}{\gamma}}$ where $\gamma$ is the learning rate. So in $T$ iterations the minimum number that the price can get to is $\frac{\sum_{\buyer \in \buyers} \budget[\buyer]}{e^{\nicefrac{-T}{\gamma}} \price[\good]^0}$. Also note, demand can be unbounded only when the price of the good is 0 otherwise demand will always be upperbounded because of the budget constraint. Since for finite learning rates the update rule can will never result in a zero price, we will always avoid unbounded prices.}
\fi
%
% It simplifies as follows:\amy{it seems like this simplification has little to do with the choice of $\delta_h$, specifically the 6. so can we present this simplification of KL/entropic t\^atonnement first -- as in in the prelim section, and then say that our choice is 6 times the KL divergence. it's just all a little muddled here.}
% \begin{align}
%     \price[\good](t+1) = \price[\good](t)\exp\left\{\frac{\excess[\good](\price(t))}{\gamma}\right\} && \text{for all } \good \in \goods, \text{ for } t = 0, 1, 2, \hdots \label{tatonnement-KL} \\
%     \price(0) \in \R^{\numgoods} \label{tatonnement-KL2} 
% \end{align}

% \noindent
% where, by abuse of notation, $\excess(\price(t))$ represents an excess demand vector at prices $\price(t)$.

\csch{The t\^atonnement process reduces to a \emph{sub\/}gradient method with fixed step size when utilities are not strictly concave, and such methods do not necessarily converge in last iterates.
Regardless, we now proceed to show that entropic t\^atonnement converges in continuous, \emph{strictly\/} concave, and homogeneous (CSCH) Fisher markets at a rate of $O(\nicefrac{1}{t})$\amy{conjecture!}.
%%%
%We use results derived by \citeauthor{grad-prop-response} \cite{grad-prop-response} and \citeauthor{ fisher-tatonnement} \cite{fisher-tatonnement} to prove a convergence rate of $O(\nicefrac{1}{t})$ for t\^atonnement in any continuous, strictly concave, and homogeneous (CSCH) Fisher market.
%%%
This result extends \citeauthor{fisher-tatonnement}'s previous worst-case rate of convergence for Leontief Fisher markets to the entire class of CSCH Fisher markets. 

%\amy{$e_{ij} / p_j$ = $i$'s demand for/allocation of good $j$}

\emph{Notation.} 
In the remainder of this section, we refer to price vector $\price(t)$ by $\price^t$, and denote $\price(t+1)- \price(t)$ by $\pricediff$. At time $t$ in the t\^atonnement process, we denote buyer $\buyer$'s Marshallian demand for good $\good$ i.e., $ x\in \marshallian[\buyer][\good](
\price(t))$, by $\demand[\buyer][\good]^t$;
the total demand for good $\good$, i.e., $\sum_{\buyer \in \buyers} \marshallian[\buyer][\good](
\price(t))$, by $\demand[\good]^t$; and
the buyer $\buyer$'s Hicksian demand for good $\good$ at utility level 1, i.e.,  $h \in \hicksian[\buyer][\good](\price(t), 1)$, by $\hicksian[\buyer][\good]^t$.

We provide a sketch of the proof used for obtaining our convergence rate in this section. The omitted lemmas and proofs can be found in Appendix \ref{sec:Proofs6}. At a high level, our proof follows \citeauthor{fisher-tatonnement}'s proof  technique for Leontief Fisher markets \cite{fisher-tatonnement}, which works as follows.
First, we prove that under certain assumptions, the condition required by \Cref{theorem-devanur} holds when $f$ is the convex potential function for CSCH Fisher markets defined in \Cref{potential-func}, i.e., $f = \potential$, the negative excess demand.
For these assumptions to be valid, we need to set $\gamma$ to be five times the maximum demand for any good throughout the t\^atonnment process.
Further, since $\gamma$ needs to be set at the outset, we need to upper bound $\gamma$.
To do so, we derive a naive as well as an informed bound on the maximum demand for any good during t\^atonnement in CCH and mixed CSCH Fisher markets, respectively.
Based on these bounds, we then derive a naive upper bound on $\gamma$ CCH Fisher markets, as well as a more informed bound for mixed CSCH Fisher markets.
Finally, we use \Cref{theorem-devanur} to obtain the convergence rate of $O(\nicefrac{1}{t})$\amy{conjecture!}. 

%\sdeni{}{\emph{We note that although our proof idea follows that of \citeauthor{fisher-tatonnement} \cite{fisher-tatonnement}, we are able to extend \citeauthor{fisher-tatonnement}'s proof ideas to all CSCH Fisher markets by exploiting our novel characterization of equilibrium prices via expenditure functions (\Cref{new-convex}) and using the duality structure between UMP and EMP from consumer theory.}}

The following lemma derives the conditions under which the antecedent of \Cref{theorem-devanur} holds for entropic t\^atonnement.
% A complete proof of \Cref{ineq-devanur} can be found in appendix $\ref{sec:Proofs4}$.

\begin{restatable}{lemma}{ineqdevanur}
\label{ineq-devanur}
The following holds for entropic t\^atonnement when run on a CSCH Fisher market $(\util, \budget)$: for all $t \in \N$:
%\begin{align}
$    \potential(\price^{t+1}) - \lapprox[\potential][\price^{t+1}][\price^t] \leq \gamma \divergence[\mathrm{h}][{\price^{t+1}}][{\price^t}]$,
%    \enspace ,
%\end{align}
%
%\noindent
where $\gamma = 5 \cdot \max\limits_{\substack{t \in \N \\ \good \in \goods}} \left\{\demand[\good]^{t} \right\}$ and $\divergence[\mathrm{h}] = 6 \cdot \divergence[\mathrm{KL}][][]$.
\end{restatable}

Combining \Cref{ineq-devanur} with \Cref{theorem-devanur}, we obtain the main result in this section, namely a worst-case convergence rate of $O(\nicefrac{1}{t})$\amy{conjecture!} for entropic t\^atonnement.

\begin{theorem}
\label{main-convergence-thm}
The following holds for entropic t\^atonnement when run on a CSCH Fisher market $(\util, \budget)$: for all $t \in \N$,
\begin{align}
    \potential(\price^t) - \potential(\price^*) \leq \frac{\gamma \divergence[\mathrm{h}][\price^*][\price^0]}{t}
    \enspace ,
\end{align}
where $\gamma = 5 \max_{\substack{\good \in \goods\\ t \in \N}} \left\{\sum_{\buyer \in \buyers}\marshallian[\buyer][\good]\left(\price(t), \budget[\buyer]\right) \right\}$ and $\divergence[\mathrm{h}] = 6 \cdot \divergence[\mathrm{KL}][][]$.
\end{theorem}

Note, however, that for the 
\Cref{main-convergence-thm} to hold, the step size $\gamma$ needs to be known in advance so that it can be set at the start of the t\^atonnement process.
In other words, we need to upper bound the demand on all goods all throughout t\^atonnement.

A naive upper bound for all CCH Fisher markets can be derived as follows.
If the number of iterations $T$ for which t\^atonnement is to be run were known at the outset, then one could use a naive bound of $\max_{t \in \N} \demand[\good]^t \leq e^{\nicefrac{T}{5}} \frac{\sum_{\buyer \in \buyers} \budget[\buyer]}{\min_{\good \in \goods} \price[\good]^0}$, since the price of a good can decrease at most by a factor of $e^{-\nicefrac{1}{5}}$ during each iteration, by \Cref{price-change}.
This time-dependent upper bound could then be combined with the doubling trick to extend the convergence result to an unknown time horizon.

In addition to this naive bound, we also provide a more informed bound on $\gamma$
%when the time horizon is not known 
for mixed CCH Fisher markets, since the naive upper bound can be quite large.
At a high level, to calculate this informed upper bound, we first bound the demand for any good in gross complements CCH Fisher markets,%
\footnote{We provide a comparison of our bound for gross complements markets to the bound provided by \citeauthor{fisher-tatonnement} \cite{fisher-tatonnement} for Leontief markets in appendix \ref{sec:Proofs6}.}
and then in gross substitutes CCH markets throughout t\^atonnement.
We then combine these upper bounds to obtain an upper bound on the demand for all goods all throughout t\^atonnement in all mixed CCH Fisher markets.
A complete proof can be found in Appendix \ref{sec:Proofs6}.

\begin{restatable}{lemma}{upperbounddemand}
\label{upper-bound-demand}
For all mixed CCH Fisher markets, if entropic t\^atonnement 
%given in \Cref{tatonnement-KL} and \Cref{tatonnement-KL2}
is run s.t.\ for all goods and for all $t \in \N$ $\frac{|\pricediff[\good]|}{\price[\good]^t} \leq \frac{1}{4}$, then the demand for any good $\good \in \goods$ throughout the process
%for all $t \in \N$ 
is bounded as follows:
\begin{align}
\max_{t \in \N} \demand[\good]^t \leq 2 \frac{\sum_{\buyer \in \buyers} \budget[\buyer]}{\price[\good]^0} \max \left\{\max_{k \in \goods}\demand[k]^0,  \frac{\max_{l \in \goods}\demand[l]^0}{\demand[\good]^0} \right\} + 2 \max_{k \in \goods} \frac{\demand[\good]^0}{  \demand[k]^0} \enspace .
\end{align}
\end{restatable}

Using \Cref{upper-bound-demand}, we then obtain the following corollary of \Cref{main-convergence-thm} which gives a worst-case convergence rate of $O(\nicefrac{1}{t})$\amy{conjecture!} for entropic t\^atonnement in mixed CSCH Fisher markets, this time with the more informed bound.
We note that if a better bound were to be derived than that given in \Cref{upper-bound-demand}, we could use it to reduce the convergence rate by a constant factor, since $\gamma$ is part of the convergence rate.
% \amy{i need an explanation of this point!}
%\enr{This is where I am getting at with my earlier comment about 'the best bound'...}\deni{Yes, I understand that but that is not what I am saying with my earlier comment, so we have to figure out a better way to state it.}\enr{ok, we talked about this!}

\begin{corollary}

\label{convergence-result}
The following holds for the entropic t\^atonnement process when run on a mixed CSCH Fisher market $(\util, \budget)$:
\begin{align}
    \potential(\price^t) - \potential(\price^*) \leq \frac{\gamma \divergence[\mathrm{h}][\price^*][\price^0]}{t}
    \enspace ,
\end{align}

\noindent
where $\gamma = 10 \cdot \max_{\good \in \goods} \left\{ \left[ \frac{\sum_{\buyer \in \buyers} \budget[\buyer]}{\price[\good]^0} \max \left\{\max_{k \in \goods}\demand[k]^0,  \frac{\max_{l \in \goods}\demand[l]^0}{\demand[\good]^0} \right\} + \max_{k \in \goods} \frac{\demand[\good]^0}{  \demand[k]^0} \right] \right\}$ and $\divergence[\mathrm{h}] = 6 \cdot \divergence[\mathrm{KL}][][]$.

\end{corollary}
}

%% file: experiments.tex
\citeauthor{fisher-tatonnement} \cite{fisher-tatonnement} proved a worst-case \emph{lower\/} bound of $\Omega(\nicefrac{1}{t^2})$ to complement their $O(\nicefrac{1}{t})$ worst-case upper bound for the convergence rate of entropic t\^atonnement in Leontief markets.
These results suggest a possible convergence rate of $O(\nicefrac{1}{t^2})$ or $O(\nicefrac{1}{t})$ for entropic t\^atonnement for a class of Fisher markets that includes Leontief markets.
%Our goals with our experiments are two-fold, to verify empirically that a sublinear convergence rate holds in CSCH Fisher markets, and 
The goal of our experiments is to better understand the class of Fisher markets for which entropic t\^atonnement converges, and to see if a worst-case convergence rate of $O(\nicefrac{1}{t^2})$ or $O(\nicefrac{1}{t})$ might hold, not only for Leontief, but for a larger class of CCH Fisher, markets.

%\csch{First, as a sanity check, we confirmed experimentally our analytically derived convergence rate.}

% \samy{Our first set of experiments excluded buyers with linear utility functions, so that the elasticity of demand for all buyers was bounded.}{}

In all our experiments, we randomly generated  mixed CES Fisher markets, each with 70 buyers and 30 goods.
The buyers' values for goods, and their budgets, were drawn uniformly between 2 and 3. We drew initial prices uniformly in the range $[2,3]$.
%\amy{this seems restrictive. shouldn't budgets be 3 or 4 or 5 times buyers' values?}
In our first two experiments, we initialized 10,000 mixed CES markets, and we chose the $\rho$ parameter uniformly at random with $\nicefrac{1}{2}$ probability in the range $[\nicefrac{1}{4}, \nicefrac{3}{4}]$ and with $\nicefrac{1}{2}$ probability in the range $ [-1, -101]$.%
\footnote{We ruled out values of $\rho$ close to 0 and 1
%\amy{it looks to me like you ruled out more than just very small values and values close to 0 and 1. 1/5, for example, is neither of those, really. same for 4/5.}
to ensure numerical stability.}
Note that this range for $\rho$ ensures that the elasticity of demand $E$ of the market is bounded above by 4.
%\amy{those initial prices seem very close to the buyers' budgets and values}\deni{they seem to be picked close to each other but because there are twice as more buyers than goods, the equilibrium prices will be in general much higher than the starting prices.}
Under these conditions, we ran the entropic t\^atonnement process with a step size of 2 in each market.
% \footnote{\sdeni{}{Note, that if the step size is picked too small then t\^atonnement might cycle.}}\amy{do you want to insert a footnote about what happens when the step size is too small? i think that's interesting/important} 
%to see if the process would converge.
%, and if so, if it ever converged at a rate worse than sublinear.

In our first set of experiments, we assigned each buyer, uniformly at random, either CES, Cobb-Douglas, or Leontief utilities, with $E \le 4$. %\amy{how was $E$ distributed?}\deni{distributed uniformly between [3,4]}
We observed convergence in all experiments, at the rate depicted in \Cref{fig:obj_change}.
These results suggest that the sublinear convergence rate of $O(\nicefrac{1}{t})$ could be improved to $O(\nicefrac{1}{t^2})$ for entropic t\^atonnement in Leontief markets, and could perhaps even be extended to a larger class of Fisher markets, beyond Leontief.
(The inner frame in \Cref{fig:obj_change} is a closeup of iterations 0 to 10, intended to highlight that the average trajectory of the objective value throughout entropic t\^atonnement decreases at a rate faster than $O(\nicefrac{1}{t^2})$.)

\wine{Second, to try to better understand the behavior of t\^atonnement in mixed CCH Fisher markets where the elasticity of demand of buyers might be unbounded, we ran almost the same experiment again, but this time, each buyer was assigned, uniformly at random, either CES, Cobb-Douglas, Leontief, or \emph{linear\/} utilities.
We then checked, for each market, if the process converged.
We show a sample entropic t\^atonnement trajectory for one mixed CCH Fisher market in \Cref{fig:obj_change}.
We see that the objective value decreases initially (at a rate slower than $\nicefrac{1}{t^2}$), but then, after about 10 iterations, it begins to oscillate.
While at times the market may be tending toward an equilibrium, it is unable to settle at one.}

We then ran the same experiment with buyers with linear utilities included (i.e., unbounded elasticity of demand), and found that out of 10,000 experiments, 9889 of them did \emph{not\/} converge.
This result is unsurprising in light of
\wineusesparingly{Example~\ref{ex:linear}}{the fact that t\^atonnement is not guaranteed to converge in linear markets}, since, in expectation, buyers with linear utilities make up a quarter of this market.

Finally, we ran experiments in which we varied the elasticity of demand. 
To do so, we ran t\^atonnement in markets with elasticities of demand $E \in \{ 0.1, 0.2, \hdots, 0.9 \}$, and we varied the step size $\gamma \in \{ 1, 2, \hdots, 9 \}$.
%to understand the behavior of entropic t\^atonnement for different step sizes.
%In particular, for each choice of elasticity of demand and step size, we initialized 20 markets, and ran t\^atonnement on all these markets with the chosen step size and recorded if these experiments converged or not.
The results are presented in \Cref{fig:heatmap}.
In this heat map, purple signifies that all experiments converged, while yellow signifies that no experiments converged.
Interestingly, as the elasticity of demand of the market increased, prices still converged, albeit only with a sufficiently large step size, \samy{and}{thus} at a slower rate.

%Since in markets with $E = \infty$, i.e., linear Fisher markets, t\^atonnement does not converge.
In the light of the results of our experiments, we conjecture that t\^atonnement converges at a rate of $O(\nicefrac{(1 + E)}{t^2})$ in CCH Fisher markets.
We recall that for Leontief utilities $E = 0$, for weak gross complements markets $E \leq 1$, for weak gross substitutes markets $E \geq 1$, and for linear utilities $E = \infty$.
Our conjecture thus implies that a convergence rate of $O(\nicefrac{1}{t^2})$ applies for Leontief Fisher markets, i.e., perfect complements, and that this rate deteriorates as the market's elasticity of demand increases, ultimately leading to non-convergence in markets of perfect substitutes, i.e., linear Fisher markets. That is, the convergence rate of t\^atonnement in CCH Fisher markets can be seen as a combination of the convergence rates of two types of extreme markets: perfect complements, i.e., Leontief, and perfect substitutes, i.e., linear, Fisher markets.

\deni{The paragraph below feels redundant.}
\sdeni{Together, these two experiments support our conjecture that in order for entropic t\^atonnement with a fixed step size to converge in CCH Fisher markets, the buyers' elasticity of demand must be bounded.
The larger the elasticity of demand is, the more responsive demand is to changes in prices.
In particular, if elasticity of demand is unbounded, then the t\^atonnement price updates can be arbitrarily large, potentially leading to cycles: high demand and low prices, followed by low demand and high prices, over and over again.
%Thus, to ensure convergence, the step size should be large enough so that t\^atonnement updates are small enough.
%
Recalling the known worst-case convergence lower bound of $\Omega(\nicefrac{1}{t^2})$ for Leontief Fisher markets \cite{fisher-tatonnement}, our experiments support the conjecture that entropic t\^atonnement may have a tight convergence bound of $\Theta(\nicefrac{1}{t^2})$ in CCH Fisher markets, assuming buyers' elasticity of demand is bounded.}{}

\begin{figure}[h]
% \centering
\begin{subfigure}[t]{0.475\textwidth}
    \centering
    \includegraphics[width=\textwidth]{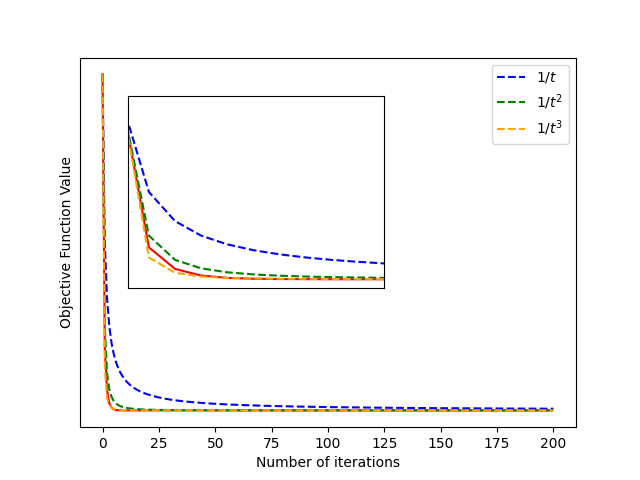}
    \caption{Average trajectory of the value of the objective function throughout t\^atonnement with KL divergence for mixed CES Fisher markets with $E \leq 4$ is drawn in {\color{red} red}. The predicted worst case sublinear convergence rate is depicted by a dashed {\color{blue} blue} line. A convergence rate of $\nicefrac{1}{t^2}$ and $\nicefrac{1}{t^3}$ are denoted in {\color{green} green} and {\color{orange} orange}, respectively.}
    \label{fig:obj_change}
\end{subfigure}
%\vline
\hspace{2em}
\begin{subfigure}[t]{0.475\textwidth}
    \centering
    \includegraphics[width=\textwidth]{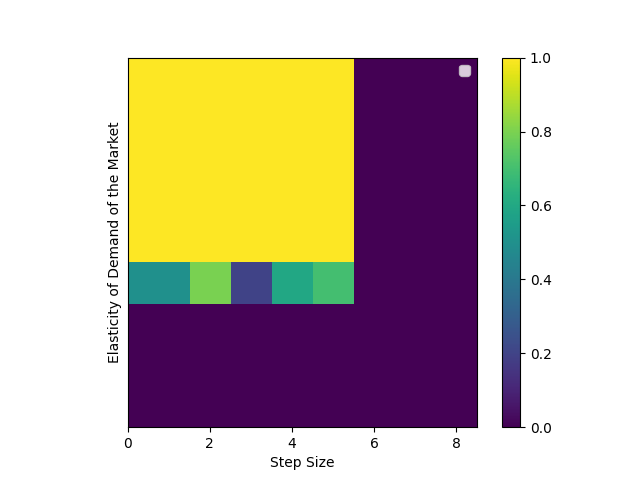}
    \caption{Percentage of experiments that converge as a function of step size and elasticity of demand. Purple signifies that all experiments converged; yellow signifies that no experiments converged.
    For sufficiently low values of $E$, we see convergence regardless of step size; and for sufficiently large step sizes, we see convergence regardless of $E$.}
    \label{fig:heatmap}
\end{subfigure}
\end{figure}

\if 0
\begin{subfigure}[t]{0.475\textwidth}
    \centering
    \includegraphics[width=\textwidth]{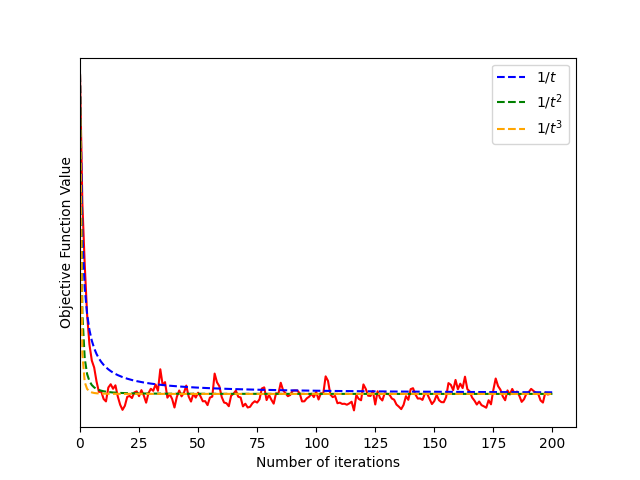}
    \caption{Example trajectory of the value of the objective function throughout t\^atonnement with KL divergence for CCH Fisher markets is drawn in {\color{red} red}. The predicted worst case sublinear convergence rate is depicted by a dashed {\color{blue} blue} line. A convergence rate of $\nicefrac{1}{t^2}$ and $\nicefrac{1}{t^3}$ are denoted in {\color{green} green} and {\color{orange} orange}, respectively.}
    \label{fig:obj_change_linear}
\end{subfigure}
\fi

%% file: conclusion.tex
\section{Conclusion}
In this paper, we introduced a new convex program 
%that generalizes the dual of the Eisenberg-Gale program
whose dual characterizes the equilibrium prices of CCH Fisher markets via expenditure functions.
We also related this dual to the dual of the Eisenberg-Gale program.
The dual of our program is easily interpretable, and thus allows us to likewise interpret the Eisenberg-Gale dual.
In particular, while it is known that an equilibrium allocation that solves the Eisenberg-Gale program (the primal) is one that maximizes
%the buyers' utilities given their budgets at equilibrium prices
the Nash social welfare, we show that equilibrium prices---the solution to the dual---minimize 
%the buyers' expenditures at the utility levels associated with their equilibrium allocations (the dual)
the distance between the sellers' surplus and the buyers' Marshallian surplus.
%as the dual goals of a fictional auctioneer who is trying to maximize the utility of each buyer constrained by their budget, while controlling prices such that the market clears. Additionally, we note that the the dual of the Eisenberg-Gale program measures the difference between the auctioneer's surplus and the collective surplus of the buyers.
Building on the results of \citeauthor{fisher-tatonnement} \cite{fisher-tatonnement}, who showed that the subdifferential of the dual of the Eisenberg-Gale program is equal to the negative excess demand, we show the same for the dual of our convex program, which implies that solving our convex program via generalized gradient descent is equivalent to solving a Fisher market by means of t\^atonnement.

The main technical innovation in this work is to express equilibrium prices via expenditure functions.
This insight could allow us to prove the convergence of t\^atonnement for more general classes of CCH utility functions, beyond CES.
To this end, we ran experiments that supported the conjecture that t\^atonnement converges at a rate of $O\left( \nicefrac{(1+ E)}{t^2}\right)$ in CCH Fisher markets with elasticity of demand bounded by $E$.
If this result holds in general, it would improve upon and generalize prior results for Leontief markets to a larger class of CCH markets, which includes nested and mixed CES utilities.
In future work, we plan to continue to investigate this conjecture, using the insights gained from our consumer-theoretic characterization of the equilibrium prices of Fisher markets.

We believe that our analysis offers important insights about the Eisenberg-Gale program.
We observe that in CCH markets, maximizing the bang-per-buck is equivalent to minimizing the buck-per-bang, and moreover, the buck-per-bang and bang-per-buck are constant across utility levels and budgets.
Additionally, optimizing prices to minimize buyers' buck-per-bang is equivalent to maximizing their utilities constrained by their budgets.
As a result, equilibrium prices can be determined by minimizing the buck-per-bang of buyers, which depends only on prices.
In other words, the computation of equilibrium prices can be decoupled from the computation of equilibrium allocations.
Indeed, there exists a primal-dual convex program for these markets.
The challenge in solving Fisher markets where buyers' utility functions can be non-homogeneous seems to stem from the fact that the buck-per-bang and bang-per-buck vary across utility levels and budget, which in turn means that the computation of prices and allocations cannot be decoupled.
As a result, we suspect that a primal-dual convex program formulation that solves Fisher markets for buyers with non-homogeneous utility functions may not exist. 

% In CCH Fisher markets, the share of spending on good $\good$, i.e., $\left( \frac{\price[\good]}{\expend[\buyer](\price, \goalutil[\buyer])} \right) \hicksian[\buyer][\good](\price, \goalutil[\buyer])$,
%  as a function of the utility level is constant.
% Likewise, the share of spending on good $\good$ across budget levels, i.e., $\left( \frac{\price[\good]}{\budget[\buyer]} \right) \marshallian[\buyer][\good](\price, \budget[\buyer])$, is constant.
% The challenge in solving Fisher markets where buyers' utility functions can be arbitrary seems to stem from the fact that for non-CCH utility functions, these spending shares can vary across utility and budget levels,
% which in turn means that the computation of prices and allocations cannot obviously be decoupled.
% As a result, we suspect that a primal-dual type convex program formulation that solves Fisher markets for buyers with non-CCH utility functions may not exist.
\wine{An interesting direction for future work would be to devise market dynamics that adjust allocations and prices together in search of equilibria.
We believe that such dynamics may be necessary to find equilibria in Fisher markets beyond CCH (e.g., continuous and concave but not necessarily homogeneous utilities), the next frontier in this line of research.
% As a first step in this direction, we provide a characterization of market dynamics that are equivalent to gradient descent on a convex potential function.
% This result can serve as a simple rule of thumb to quickly determine whether the convergence rate of a market dynamic can be analyzed using known results about generalized gradient descent.
%
Related, to the best of our knowledge, Marshallian consumer surplus in Fisher markets is not well understood. On the contrary, Marshallian consumer surplus is mostly studied in markets with a unique good, and other than \citeauthor{vives1987marshallian} \cite{vives1987marshallian}, not much effort has been put into obtaining explicit expressions for Marshallian consumer surplus. In our work, we have shown that in CCH Fisher markets, the Marshallian consumer surplus can be expressed as a function of the expenditure of the buyers.
It remains to be seen if a similar explicit expression can be obtained for utility functions beyond CCH.
We believe that future work aimed at understanding the Marshallian consumer surplus in Fisher markets could further our understanding of these markets, perhaps beyond CCH Fisher markets.}{}

%% file: acks.tex
\section*{Acknowledgments}

We would like to thank Richard Cole, Yun Kuen Cheung, and Yixin Tao for feedback on an earlier version of this paper. This work was partially supported by NSF Grant CMMI-1761546.

%% file: appendix/proofs_sec3.tex
% \equivoptim*

Recall that the dual proposed by \citeauthor{cole2019balancing} \cite{cole2019balancing} is given by: 
\begin{align*}
    \min_{\price \in \R^\numgoods} \sum_{\good \in \goods} \price[\good] + \sum_{\buyer \in \buyers} \budget[\buyer] \log\left(\max_{\allocation[\buyer]\in \R^{\numgoods}_+: \allocation[\buyer] \cdot \price \leq \budget[\buyer]}\util[\buyer](\allocation[\buyer])\right)
\end{align*}

This dual's optimal differs from the optimal value of the Eisenberg-Gale program by a constant factor (of $\sum_{\buyer \in \buyers} \budget[\buyer]$) as shown by the following example:

\begin{example}\label{dual-diff-cole}
Consider a linear Fisher market with only one good and one buyer with a utility of $1$ for the good and a budget of $1$ as well.
The equilibrium of this market is given by $\allocation[1][1]^* = 1, \price[1]^* = 1$. 
The primal of the Eisenberg-Gale program thus evaluates to $\budget[1]\log(\allocation[1][1]^*) = (1) \log(1) = 0$, while the dual given by \citeauthor{cole2019balancing} evaluates to $1 \log(1) + 1 = 1$.
Hence, the optimal primal value is not equal to the optimal dual value of the dual given by \citeauthor{cole2019balancing}, so this dual is not exactly the dual of the Eisenberg-Gale program.
\end{example}

We now derive the dual of the Eisenberg-Gale program.
We begin with an essential lemma, which states that UMP for CCH utility functions can be expressed as an unconstrained optimization problem.
% A proof of this lemma can be found in \Cref{sec:Proofs3}. 

\begin{restatable}{lemma}{equivoptim}
\label{equiv-optim}
The optimization problem
\begin{align}
\max_{\allocation[\buyer] \in \R^\numgoods_+: \allocation[\buyer] \cdot \price \leq \budget[\buyer]} \budget[\buyer] \log (\util[\buyer](\allocation[\buyer]))
\end{align}
is equivalent to the optimization problem
\begin{align}
    \max_{\allocation[\buyer] \in \R_+^\numgoods} \left\{\budget[\buyer] \log(\util[\buyer](\allocation[\buyer])) + \budget[\buyer] - \allocation[\buyer] \cdot \price \right\}
\enspace .
\end{align}
\end{restatable}

\begin{proof}[\Cref{equiv-optim}]
The Lagrangian associated with $\max_{\allocation[\buyer]: \allocation[\buyer] \cdot \price \leq \budget[\buyer]} \budget[\buyer] \log (\util[\buyer](\allocation[\buyer]))$ is given by:
\begin{align*}
    L(\allocation[
    \buyer], \lambda, \bm{\mu}) = \budget[\buyer] \log( \util[\buyer] (\allocation[\buyer])) + \lambda \left(  \budget[\buyer] - \allocation[\buyer] \cdot \price \right) + \bm{\mu}^T \allocation[\buyer] \enspace ,
\end{align*}

\noindent
where $\lambda \in \R_+$ and $\bm \mu \in \R^\numgoods_+$ are slack variables.

Let $(\allocation[\buyer]^*, \lambda^*, \bm \mu^*)$ be an optimal solution to the Lagrangian.
From the KKT stationarity condition for this Lagrangian  \cite{kuhn1951kkt}, it holds that, for all $\good \in \goods$,
\begin{align*}
    \frac{\budget[\buyer]}{\util[\buyer] (\allocation[\buyer]^*)} \left[ \frac{\partial \util[\buyer]}{\partial \allocation[\buyer][\good]} \right]_{\allocation[\buyer] = \allocation[\buyer]^*} - \lambda^* \price[\good] + \mu_{\good}^* \doteq 0 \\
    \frac{\budget[\buyer]}{\util[\buyer] (\allocation[\buyer]^*)} \left[ \frac{\partial \util[\buyer]}{\partial \allocation[\buyer][\good]} \right]_{\allocation[\buyer] = \allocation[\buyer]^*} \allocation[\buyer][\good]^* - \lambda^* \price[\good] \allocation[\buyer][\good]^* + \mu_{\good}^* \allocation[\buyer][\good]^* = 0 \\
    \frac{\budget[\buyer]}{\util[\buyer] (\allocation[\buyer]^*)} \left[ \frac{\partial \util[\buyer]}{\partial \allocation[\buyer][\good]} \right]_{\allocation[\buyer] = \allocation[\buyer]^*} \allocation[\buyer][\good]^* - \lambda^* \price[\good] \allocation[\buyer][\good]^* = 0
    \enspace .
\end{align*}
\noindent
The penultimate line is obtained by multiplying both sides by $\allocation[\buyer][\good]^*$, and the last line, by the KKT complementarity condition, namely $\mu_{\good}^* \allocation[\buyer][\good]^* = 0$.

Summing up across all $\good \in \goods$ on both sides yields:
\begin{align*}
    \frac{\budget[\buyer]}{\util[\buyer](\allocation[\buyer]^*)} \sum_{\good \in \goods}\left[\frac{\partial \util[\buyer]}{\partial \allocation[\buyer][\good]}\right]_{\allocation[\buyer] = \allocation[\buyer]^*}\allocation[\buyer][\good]^* - \lambda^* \sum_{\good \in \goods} \price[\good] \allocation[\buyer][\good]^*  = 0\\
    \frac{\budget[\buyer]}{\util[\buyer](\allocation[\buyer]^*)}\util[\buyer](\allocation[\buyer]^*) - \lambda^* \sum_{\good \in \goods} \price[\good] \allocation[\buyer][\good]^*  = 0\\
    \budget[\buyer] -\lambda^* \budget[\buyer]  = 0\\
    \lambda^* = 1 \enspace ,
\end{align*}

\noindent
where the second line is obtained from Euler's theorem for homogeneous functions \cite{lewis1969homogeneous}, and the last line, from the KKT complementarity condition again, namely $\lambda^* \left(\sum_{\good \in \goods} \budget[\buyer] - \price[\good] \allocation[\buyer][\good]^* \right) = 0$.

Hence, plugging $\lambda^* = 1$
back into the Lagrangian restricted to $\R^\numgoods_+$, we get:
\begin{align*}
    \max_{\allocation[\buyer] \in \R^\numgoods_+: \allocation[\buyer] \cdot \price \leq \budget[\buyer]} \budget[\buyer] \log (\util[\buyer] (\allocation[\buyer])) &= \max_{\allocation[\buyer] \in \R^\numgoods_+} \budget[\buyer] \log( \util[\buyer] (\allocation[\buyer])) + \lambda^* \left( \budget[\buyer] -   \allocation[\buyer] \cdot \price \right)\\
    &= \max_{\allocation[\buyer] \in \R^\numgoods_+} \budget[\buyer] \log(\util[\buyer] (\allocation[\buyer])) + \budget[\buyer] - \allocation[\buyer] \cdot \price  \enspace .
\end{align*}
\end{proof}

With this lemma in hand, we can now derive the dual of the Eisenberg-Gale program.

\EGgeneralization*

\begin{proof}[\Cref{EG-generalization}]
The Lagrangian dual function $g: \R^\numgoods \to \R$ of the Eisenberg-Gale primal is given by:
\begin{align*}
    g(\price) 
    &= \max_{\allocation \in \R^{\numbuyers \times \numgoods}_+} L(\allocation, \price) \\
    &= \max_{\allocation \in \R^{\numbuyers \times \numgoods}_+} \left\{ \sum_{\buyer \in \buyers} \budget[\buyer] \log( \util[\buyer] (\allocation[\buyer])) + \sum_{\good \in \goods} \price[\good] \left( 1 - \sum_{\buyer \in \buyers} \allocation[\buyer][\good] \right) \right\} \\
    %&= \sum_{\good \in \goods} \price[\good] + \max_{\allocation \in \R^{\numbuyers \times \numgoods}_+} \left\{ \sum_{\buyer \in \buyers} \budget[\buyer] \log( \util[\buyer] (\allocation[\buyer])) - \sum_{\good \in \goods} \sum_{\buyer \in \buyers} \price[\good] \allocation[\buyer][\good] \right\} \\
    &= \sum_{\good \in \goods} \price[\good] +  \max_{\allocation \in \R^{\numbuyers \times \numgoods}_+} \left\{ \sum_{\buyer \in \buyers} \left( \budget[\buyer] \log( \util[\buyer] (\allocation[\buyer])) - \sum_{\good \in \goods} \price[\good] \allocation[\buyer][\good] \right)\right\} \\
    &= \sum_{\good \in \goods} \price[\good] + \sum_{\buyer \in \buyers} \max_{\allocation[\buyer] \in \R^{ \numgoods}_+} \left\{\budget[\buyer] \log(\util[\buyer] (\allocation[\buyer])) - \sum_{\good \in \goods} \price[\good] \allocation[\buyer][\good] \right\} \\
    %&= \sum_{\good \in \goods} \price[\good] + \sum_{\buyer \in \buyers} \max_{\allocation \in \R^{\numbuyers \times \numgoods}_+} \left\{ \budget[\buyer] \log( \util[\buyer] (\allocation[\buyer])) + \budget[\buyer] - \budget[\buyer] - \price \cdot \allocation[\buyer] \right\} \\
    &= \sum_{\good \in \goods} \price[\good] + \sum_{\buyer \in \buyers} \left[ \max_{\allocation[\buyer] \in \R^{ \numgoods}_+} \left\{ \budget[\buyer] \log( \util[\buyer] (\allocation[\buyer])) + \budget[\buyer] - \price \cdot \allocation[\buyer] \right\} - \budget[\buyer] \right] \\
    &= \sum_{\good \in \goods} \price[\good] + \sum_{\buyer \in \buyers} \left[ \max_{\allocation[\buyer] \in \R^\numgoods_+: \allocation[\buyer] \cdot \price \leq \budget[\buyer]} \budget[\buyer] \log (\util[\buyer] (\allocation[\buyer])) - \budget[\buyer]\right] && \text{(\Cref{equiv-optim})} \\
    &= \sum_{\good \in \goods} \price[\good] + \sum_{\buyer \in \buyers} \left[ \budget[\buyer] \log \left(\max_{\allocation[\buyer] \in \R^\numgoods_+: \allocation[\buyer] \cdot \price \leq \budget[\buyer]} \util[\buyer](\allocation[\buyer]) \right) - \budget[\buyer]\right] \\
    &= \sum_{\good \in \goods} \price[\good] + \sum_{\buyer \in \buyers} \left( \budget[\buyer] \log \left( \indirectutil[\buyer] (\price, \budget[\buyer]) \right) - \budget[\buyer] \right)
\end{align*}

\noindent
The order of the $\max$ and the sum over all buyers can be interchanged in this proof because prices are given, which renders the maximization problem for buyer $\buyer$ independent of that of buyer $\buyer'$. 
Therefore, the Eisenberg-Gale dual is $\min_{\price \in \R^\numgoods_+} g(\price) = \min_{\price \in \R^\numgoods_+} \sum_{\good \in \goods} \price[\good] + \sum_{\buyer \in \buyers} \left( \budget[\buyer] \log \left(\indirectutil[\buyer](\price, \budget[\buyer])\right) - \budget[\buyer] \right)$.
\end{proof}

%% file: appendix/proofs_sec4.tex
\begin{lemma}
\label{homo-expend}
\label{homo-indirect-util}
Suppose that $\util[\buyer]$ is homogeneous, i.e., $\forall \lambda > 0, \util[\buyer](\lambda \allocation[\buyer]) = \lambda \util[\buyer]( \allocation[\buyer])$.
Then, the expenditure function and the Hicksian demand are homogeneous in $\goalutil[\buyer]$, i.e., for all $\forall \lambda > 0$, $ \expend[\buyer](\price, \lambda \goalutil[\buyer]) = \lambda \expend[\buyer](\price,  \goalutil[\buyer])$ and $\hicksian[\buyer](\price, \lambda \goalutil[\buyer]) = \lambda \hicksian[\buyer](\price, \goalutil[\buyer])$.
Likewise, the indirect utility function and the Marshallian demand are homogeneous in $\budget[\buyer]$, i.e., for all $\forall \lambda > 0$, $\indirectutil[\buyer](\price, \lambda \budget[\buyer]) = \lambda \indirectutil[\buyer](\price, \budget[\buyer])$ and $\marshallian[\buyer](\price, \lambda \budget[\buyer]) = \lambda \marshallian[\buyer](\price, \budget[\buyer])$.
\end{lemma}

\begin{proof}[\Cref{homo-expend}]
Without loss of generality, assume $\util[\buyer]$ is homogeneous of degree 1.%
\footnote{If the utility function is homogeneous of degree $k$, we can use a monotonic transformation, namely take the $k^{th}$ root, to transform the utility function into one of degree 1, while still preserving the preferences that it represents.}

For Hicksian demand, we have that:
\begin{align}
    &\hicksian[\buyer](\price, \lambda \goalutil[\buyer]) \\
    &= \argmin_{\allocation[\buyer]: \util[\buyer](\allocation[\buyer]) \geq \lambda \goalutil[\buyer]} \price \cdot \left(\lambda \frac{\allocation[\buyer]}{\lambda} \right) \\
    &= \lambda \argmin_{\allocation[\buyer]:  \util[\buyer](\frac{\allocation[\buyer]}{\lambda} ) \geq \goalutil[\buyer]} \price \cdot \left( \frac{\allocation[\buyer]}{\lambda} \right) \\
    &= \argmin_{\allocation[\buyer]: \util[\buyer] \left(\allocation[\buyer] \right) \geq \goalutil[\buyer]} \price \cdot \allocation[\buyer] \\
    &= \lambda \hicksian[\buyer](\price, \goalutil[\buyer]) \enspace . \label{eq:homo-expend}
\end{align}

\noindent 
The first equality follows from the definition of Hicksian demand;
the second, by the homogeneity of $\util[\buyer]$;
the third, by the nature of constrained optimization; and the last, from the definition of Hicksian demand again. 
% Substituting $\bm{\chi_{\buyer}} = \frac{1}{\lambda} \allocation[\buyer]$, we obtain the homogeneity of Hicksian demand:
% \begin{align}
%     \hicksian[\buyer](\price, \lambda \goalutil[\buyer]) = \argmin_{\bm \chi_{\buyer}: \util[\buyer](\bm \chi_{\buyer}) \geq \goalutil[\buyer]} \price \cdot \lambda \bm \chi_{\buyer} 
%     = \lambda \argmin_{\bm \chi_{\buyer}: \util[\buyer](\chi_{\buyer}) \geq \goalutil[\buyer]} \price \cdot \bm \chi_{\buyer}
%     = \lambda \hicksian[\buyer](\price, \goalutil[\buyer]) \label{eq:homo-expend}
% \end{align}
%
% \noindent
% where the last equality follows from the definition of the Hicksian demand.
This result implies homogeneity of the expenditure function in $\goalutil[\buyer]$:
\begin{align*}
    \expend[\buyer](\price, \lambda \goalutil[\buyer]) = \hicksian[\buyer](\price, \lambda \goalutil[\buyer]) \cdot \price 
    = \lambda \hicksian[\buyer](\price, \goalutil[\buyer]) \cdot \price 
    = \lambda \expend[\buyer](\price, \goalutil[\buyer]) \enspace .
\end{align*}

\noindent
The first and last equalities follow from the definition of the expenditure function, while the second equality follows from the homogeneity of Hicksian demand (\Cref{eq:homo-expend}).

The proof in the case of Marshallian demand and the indirect utility function is analogous.
\end{proof}

\deriveexpend*
\begin{proof}[\Cref{deriv-indirect-util}]
We prove differentiability from first principles: 
\begin{align*}
    \lim_{h \to 0} \frac{\expend[\buyer](\price, \goalutil[\buyer]+h) - \expend[\buyer](\price, \goalutil[\buyer])}{h} &= \lim_{h \to 0} \frac{\expend[\buyer](\price, (1)(\goalutil[\buyer]+h)) - \expend[\buyer](\price, (1)\goalutil[\buyer])}{h} \\
    &= \lim_{h \to 0} \frac{\expend[\buyer](\price, 1)(\goalutil[\buyer] + h) - \expend[\buyer](\price, 1)(\goalutil[\buyer])}{h} \\
    &= \lim_{h \to 0} \frac{\expend[\buyer](\price, 1)(\goalutil[\buyer] + h - \goalutil[\buyer])}{h} \\
    &= \lim_{h \to 0} \frac{\expend[\buyer](\price, 1)(h)}{h} \\
    &= \expend[\buyer](\price, 1)
\end{align*}

\noindent
The first line follows from the definition of the derivative; the second line, by homogeneity of the expenditure function (\Cref{homo-expend}), since $\util[\buyer]$ is homogeneous; and the final line follows from the properties of limits.
The other two lines follow by simple algebra.

Hence, as $\expend[\buyer](\price, \goalutil[\buyer])$ is differentiable in $\goalutil[\buyer]$, its subdifferential is a singleton with $\subdiff[{\goalutil[\buyer]}] \expend[\buyer](\price, \goalutil[\buyer]) = \left\{\expend[\buyer](\price, 1) \right\}$.
The proof of the analogous result for the indirect utility function's derivative with respect to $\budget[\buyer]$ is similar .
\end{proof}

\inverseexpend*
\begin{proof}[\Cref{inverse-expend}]
By \Cref{derive-expend}, we know that $\expend[\buyer](\price, \goalutil[\buyer])$ is differentiable in $\goalutil[\buyer]$ and that $\subdiff[{\goalutil[\buyer]}] \expend[\buyer](\price, \goalutil[\buyer]) = \left\{ \expend[\buyer](\price, 1) \right\}$. Similarly, by \Cref{deriv-indirect-util}, we know that $\subdiff[{\budget[\buyer]}] \indirectutil[\buyer](\price, \budget[\buyer])$ is differentiable in $\budget[\buyer]$ and that $\subdiff[{\budget[\buyer]}] \indirectutil[\buyer](\price, \budget[\buyer]) = \left\{ \indirectutil[\buyer](\price, 1) \right\}$. Combining these facts yields:
\begin{align*}
    \subdiff[{\goalutil[\buyer]}] \expend[\buyer](\price, \goalutil[\buyer]) \cdot \subdiff[{\budget[\buyer]}] \indirectutil[\buyer](\price, \budget[\buyer]) &=  \expend[\buyer](\price, 1) \cdot  \indirectutil[\buyer](\price, 1)  && \text{(\Cref{deriv-indirect-util})} \\
    % &=  \expend[\buyer](\price, 1) \indirectutil[\buyer](\price, 1)\\
    &= \expend[\buyer](\price, \indirectutil[\buyer](\price, 1))  && \text{(\Cref{homo-expend})} \\
    &=  1  && \text{(\Cref{expend-to-budget})}
\end{align*}

\noindent
Therefore, $\frac{1}{\partderiv[{\expend[\buyer](\price, \goalutil[\buyer])}][{\goalutil[\buyer]}]} = \partderiv[{\indirectutil[\buyer](\price, \budget[\buyer])}][{\budget[\buyer]}]$. Combining this conclusion with \Cref{deriv-indirect-util}, we obtain the result.
\end{proof}

\begin{lemma}
\label{differ-const}
Given a CCH Fisher market $(\util, \budget)$,
the dual of our convex program (\Cref{new-convex}) and that of Eisenberg Gale differ by a constant, namely $\sum_{\buyer \in \buyers} \left( \budget[\buyer] \log \budget[\buyer] -  \budget[\buyer] \right)$.
In particular,
\begin{align*}
&\min_{\price \in \R^\numgoods_+} \left\{\sum_{\good \in \goods} \price[\good] - \sum_{\buyer \in \buyers} \budget[\buyer] \log{ \left( \partial_{\goalutil[\buyer]}{\expend[\buyer] (\price, \goalutil[\buyer])} \right)} \right\}\\
&=
\min_{\price \in \R^\numgoods_+} \sum_{\good \in \goods} \price[\good] + \sum_{\buyer \in \buyers} \left( \budget[\buyer] \log{ \left( \indirectutil[\buyer](\price, \budget[\buyer]) \right)} - \budget[\buyer] \right)
- \sum_{\buyer \in \buyers} \left(  \budget[\buyer] \log \budget[\buyer] - \budget[\buyer] \right)
\end{align*}
\end{lemma}

\begin{proof}[\Cref{differ-const}]
\begin{align*}
    % &\min_{\price \in \R^\numgoods_+} \sum_{\good \in \goods} \price[\good] +\sum_{\buyer \in \buyers} \budget[\buyer] \log{ \left( \max_{\allocation[\buyer] \in \R^\numgoods_+: \price \cdot \allocation[\buyer] \leq \budget[\buyer]} \util[\buyer](\allocation[\buyer]) \right)} - \sum_{\buyer \in \buyers} \budget[\buyer] \\
    &\min_{\price \in \R^\numgoods_+} \sum_{\good \in \goods} \price[\good] + \sum_{\buyer \in \buyers} \left( \budget[\buyer] \log{ \left( \indirectutil[\buyer](\price, \budget[\buyer]) \right)} - \budget[\buyer] \right) \\
    &= \min_{\price \in \R^\numgoods_+} \sum_{\good \in \goods} \price[\good] + \sum_{\buyer \in \buyers} \budget[\buyer] \log{ \left(\budget[\buyer] \indirectutil[\buyer](\price, 1) \right)} - \sum_{\buyer \in \buyers} \budget[\buyer] && \text{(\Cref{homo-indirect-util})} \\
    % &= \min_{\price \in \R^\numgoods_+} \sum_{\good \in \goods} \price[\good] + \sum_{\buyer \in \buyers} \left( \budget[\buyer] \log{ \left( \indirectutil[\buyer](\price, 1) \right)} + \budget[\buyer] \log \budget[\buyer] \right) - \sum_{\buyer \in \buyers} \budget[\buyer] \\
    % &= \min_{\price \in \R^\numgoods_+} \sum_{\good \in \goods} \price[\good] + \sum_{\buyer \in \buyers} \budget[\buyer] \log{ \left( \indirectutil[\buyer](\price, 1) \right)} + \sum_{\buyer \in \buyers} \budget[\buyer]\log \budget[\buyer]  \amy{same as line below? delete} \\
    &= \min_{\price \in \R^\numgoods_+} \left\{\sum_{\good \in \goods} \price[\good] + \sum_{\buyer \in \buyers}  \budget[\buyer] \log{ \left( \indirectutil[\buyer](\price, 1) \right)} \right\} + \sum_{\buyer \in \buyers} \budget[\buyer]\log \budget[\buyer] - \sum_{\buyer \in \buyers} \budget[\buyer] && \label{eq:indep-of--budget} \\
    &= \min_{\price \in \R^\numgoods_+} \left\{\sum_{\good \in \goods} \price[\good] - \sum_{\buyer \in \buyers}  \budget[\buyer] \log{ \left( \frac{1}{\indirectutil[\buyer](\price, 1)} \right)} \right\} + \sum_{\buyer \in \buyers} \budget[\buyer]\log \budget[\buyer] - \sum_{\buyer \in \buyers} \budget[\buyer] \\
    &= \min_{\price \in \R^\numgoods_+} \left\{\sum_{\good \in \goods} \price[\good] - \sum_{\buyer \in \buyers}  \budget[\buyer] \log{ \left( \expend[\buyer](\price, 1) \right)} \right\} + \sum_{\buyer \in \buyers} \budget[\buyer]\log \budget[\buyer] - \sum_{\buyer \in \buyers} \budget[\buyer] && \text{(\Cref{inverse-expend})} \\
    &= \min_{\price \in \R^\numgoods_+} \left\{\sum_{\good \in \goods} \price[\good] - \sum_{\buyer \in \buyers} \budget[\buyer] \log{ \left( \partial_{\goalutil[\buyer]}{\expend[\buyer] (\price, \goalutil[\buyer])} \right)} \right\} + \sum_{\buyer \in \buyers} \budget[\buyer]\log \budget[\buyer] - \sum_{\buyer \in \buyers} \budget[\buyer] && \text{(\Cref{derive-expend})}
\end{align*}
\end{proof}

\newconvex*

\begin{proof}[\Cref{new-convex}]
By Lemma~\ref{differ-const}, our dual and the Eisenberg-Gale dual differ by a constant, which is
%$\sum_{\buyer \in \buyers} \budget[\buyer]\log \budget[\buyer] - \sum_{\buyer \in \buyers} \budget[\buyer]$
independent of the decision variables $\price \in \R^\numgoods_+$.
Hence, the optimal prices $\price^*$ of our dual are the same as those of the Eisenberg-Gale dual, and thus correspond to equilibrium prices in the CCH Fisher market $(\util, \budget)$.
Finally, 
%since our dual differs from Eisenberg-Gale's by a constant factor, namely $\sum_{\buyer \in \buyers} \left( \budget[\buyer] \log \budget[\buyer] - \budget[\buyer] \right)$,
the objective function of our convex program's primal is:
\begin{align*}
    \sum_{\buyer \in \buyers}  \budget[\buyer] \log{\left(\util[\buyer] \left( \allocation[\buyer]\right) \right)} - \sum_{\buyer \in \buyers} \left( \budget[\buyer] \log \budget[\buyer] - \budget[\buyer] \right)
    = \sum_{\buyer \in \buyers}  \budget[\buyer] \log{\util[\buyer] \left( \frac{\allocation[\buyer]}{\budget[\buyer ]}\right)} + \sum_{\buyer \in \buyers} \budget[\buyer] \enspace .
\end{align*}
%giving the desired result.
\end{proof}

%% file: appendix/proofs_sec5.tex
Danskin's theorem \cite{danskin1966thm} offers insights into optimization problems of the form:
%\begin{align}
%\label{Danskins-thm-format}
    $\min_{\x \in X} f (\x, \price)$,
%\enspace ,
%\end{align}
where $X \subset \R^\numgoods$ is compact and non-empty.
Among other things, Danskin's theorem allows us to compute the subdifferential of value of this optimization problem with respect to $\price$.
\begin{theorem}[Danskin's Theorem \cite{danskin1966thm}]
\label{Danskinsthm}
    Consider an optimization problem of the form:
    $\min_{\x \in X} f (\x, \price)$, where $X \subset \R^\numgoods$ is compact and non-empty.
    Suppose that $X$ is convex and that $f$ is concave in $\x$. Let $V(\price) = \min_{\x \in X} f (\x, \price)$ and 
    % \amy{$\inners$?}
    $X^*(\price) = \argmin_{\x \in X} f (\x, \price)$. Then the subdifferential of $V$ at $\widehat{\price}$ is given by
    %\begin{align}
        $\subdiff[\price] V(\widehat{\price}) = \left\{  \grad[\price] f (\x^*(\widehat{\price}), \widehat{\price}) \mid \x^*(\widehat{\price}) \in X^*(\widehat{\price})\right\}$ .
    %\end{align}
\end{theorem}

\shepherd*
\begin{proof}[\cref{shepherd}]
Recall that $\expend[\buyer](\price, \goalutil[\buyer]) = \min_{\x \in \R^\numgoods_+ : \util[\buyer](\x) \geq \goalutil[\buyer]} \price \cdot \x$.
Without loss of generality, we can assume that consumption set is bounded from above, since utilities are assumed to represent locally non-satiated preferences, i.e., $ \min_{\x \in X : \util[\buyer](\x) \geq \goalutil[\buyer]} \price \cdot \x$ where $X \subset \R^\numgoods_+$ is compact.
Using Danskin's theorem:
\begin{align*}
    \subdiff[\price] \expend[\buyer](\price, \goalutil[\buyer]) &= \left\{  \grad[\price] \left( \price \cdot \x \right)(\x^*(\price, \goalutil[\buyer])) \mid \x^*(\price, \goalutil[\buyer]) \in \hicksian[\buyer](\price, \goalutil[\buyer])\right\} && \text{(Danskin's Thm)} \\
    &= \left\{  \x^*(\price, \goalutil[\buyer]) \mid \x^*(\price, \goalutil[\buyer]) \in \hicksian[\buyer](\price, \goalutil[\buyer])\right\} \\
    &= \hicksian[\buyer](\price, \goalutil[\buyer])
\end{align*}

\noindent
The first equality follows from Danskin's theorem, using the facts that the objective of the expenditure minimization problem is affine and the constraint set is compact.
The second equality follows by calculus, and the third, by the definition of Hicksian demand.
\end{proof}

\excessdemand*
\begin{proof}[\Cref{excess-demand}]
% \sdeni{Since the utility functions are continuous, concave, and represent locally non-satiated preferences, $\expend[\buyer](\price, \goalutil[\buyer])$ is continuous and concave in $\price$ \cite{levin-notes, mas-colell}. Moreover, the derivative of the expenditure function with respect to the utility level, namely $\expend[\buyer](\price, 1)$, is also continuous and concave.
% Since the logarithm is also a concave function, and the composition of two concave functions remains concave,
% $-\sum_{\buyer \in \buyers} \budget[\buyer] \log{ \left(\partderiv[{\expend[\buyer](\price, \goalutil[\buyer])}][{\goalutil[\buyer]}] \right)}$ is convex, ensuring that the objective function is also convex, since $\sum_{\good \in \goods} \price[\good]$ is affine.
% Hence, the minimum of this objective function is well-defined.

% The first-order optimality condition of our convex program is that the subdifferential of the objective function contains 0 at an optimal solution $\price^*$.
% But then,}{} 
For all goods $\good \in \goods$, we have:
\begin{align*}
    &\subdiff[{\price[\good]}] \left(\sum_{\good \in \goods} \price[\good] - \sum_{\buyer \in \buyers} \budget[\buyer] \log{ \partial_{\goalutil[\buyer]}{\expend[\buyer] (\price, \goalutil[\buyer])}} \right) \\
    &=\{1\} - \subdiff[{\price[\good]}] \left(\sum_{\buyer \in \buyers} \budget[\buyer] \log{ \partial_{\goalutil[\buyer]}{\expend[\buyer] (\price, \goalutil[\buyer])}} \right) \\
    &= \{1\} - \sum_{\buyer \in \buyers} \subdiff[{\price[\good]}] \left( \budget[\buyer] \log{ \partial_{\goalutil[\buyer]}{\expend[\buyer] (\price, \goalutil[\buyer])}} \right) \\
    &=\{1\} - \sum_{\buyer \in \buyers} \marshallian[\buyer][\good](\price, \budget[\buyer]) && \text{(\Cref{lemma-deriv-marshallian})} \\
    &= -\excess[\good](\price)
\end{align*}

% \noindent
% In other words, all of good $\good$ is demanded by the buyers in aggregate, so the market clears.
% Moreover, by the definition of the Marshallian demand, each buyer maximizes their utility constrained by their budget.
% Therefore, the solution $\price^*$ of our convex program is an equilibrium price vector of the Fisher market $(\util, \budget)$.
\end{proof}

\lemmaderivmarshallian*

\begin{proof}[\Cref{lemma-deriv-marshallian}]
Without loss of generality, we can assume $\util[\buyer]$ is homogeneous of degree 1. Then:
\begin{align*}
    \subdiff[\price] \left( 
    \budget[\buyer] \log \left(\partderiv[{\expend[\buyer](\price, \goalutil[\buyer])}][{\goalutil[\buyer]}] \right) \right)
    &= \left( \frac{\budget[\buyer]}{\partderiv[{\expend[\buyer](\price, \goalutil[\buyer])}][{\goalutil[\buyer]}]} \right) \subdiff[\price] \left( 
    \partderiv[{\expend[\buyer](\price, \goalutil[\buyer])}][{\goalutil[\buyer]}] \right) 
    % && \text{(Derivative of the logarithm)} 
    \\
    &= \budget[\buyer] \left( \partderiv[{\indirectutil[\buyer](\price, \budget[\buyer])}][{\budget[\buyer]}] \right) \subdiff[\price] \left( \partderiv[{\expend[\buyer](\price, \goalutil[\buyer])}][{\goalutil[\buyer]}] \right) && \text{(\Cref{inverse-expend})} \\
    &= \budget[\buyer] \left( \partderiv[{\indirectutil[\buyer](\price, \budget[\buyer])}][{\budget[\buyer]}] \right) \subdiff[\price]
    \expend[\buyer](\price, 1) && \text{(\Cref{derive-expend})} \\
    &= \budget[\buyer] \left( \partderiv[{\indirectutil[\buyer](\price, \budget[\buyer])}][{\budget[\buyer]}] \right) \hicksian[\buyer](\price, 1) && \text{(Shephard's Lemma)} \\
    &= \budget[\buyer] \, \indirectutil[\buyer](\price, 1) \, \hicksian[\buyer](\price, 1) && \text{(\Cref{deriv-indirect-util})} \\
    &=  \indirectutil[\buyer](\price, \budget[\buyer]) \, \hicksian[\buyer](\price, 1) && \text{(\Cref{homo-indirect-util})} \\
    &=  \hicksian[\buyer]\left(\price, \indirectutil[\buyer](\price, \budget[\buyer]) \right) && \text{(\Cref{homo-expend})} \\
    &=  \marshallian[\buyer](\price, \budget[\buyer]) && \text{(\Cref{hicksian-marshallian})}
\end{align*}
\end{proof}